\documentclass[12pt]{amsart} 
\usepackage[margin=1.3in]{geometry}
\usepackage{graphicx}
\usepackage{multirow}
\usepackage[table,xcdraw]{xcolor}
\usepackage{graphicx}
\usepackage[numbers]{natbib}
\usepackage{amsmath}
\usepackage{amssymb}
\usepackage{epstopdf}
\usepackage{pdfpages}
\usepackage{amsthm}
\usepackage{xcolor}

\newtheorem{theorem}{Theorem}[section]

\newtheorem{corollary}[theorem]{Corollary}

\newtheorem{remark}[theorem]{Remark}

\newcommand*\Prob{\mathop{}\!\mathbb{P}}

\newcommand*\ee{\mathop{}\!\operatorname{e}}
\newcommand*\dd{\mathop{}\!\mathrm{d}}
\usepackage{multicol}

\author[M. \smash{Bladt}]{Martin Bladt}
\address[M. Bladt]{Department of Actuarial Science, Faculty of Business and Economics, University of Lausanne, CH-1015 Lausanne, Switzerland}

\email{{martin.bladt@unil.ch}}

\author[H. \smash{Albrecher}]{Hansj\"org Albrecher}
\address[H. Albrecher]{ Department of Actuarial Science, Faculty of Business and Economics and Swiss Finance Institute, University of Lausanne, CH-1015 Lausanne, Switzerland}

\email{{hansjoerg.albrecher@unil.ch}}

\author[J. \smash{Beirlant}]{Jan Beirlant}
\address[J. Beirlant]{Department of Mathematics,
	KU Leuven, Celestijnenlaan 200B, B-3001 Leuven, Belgium}

\email{{jan.beirlant@kuleuven.be}}

\title{Combined Tail Estimation Using Censored Data and Expert Information}

\begin{document}

\begin{abstract}
We study tail estimation in Pareto-like settings for datasets with a high percentage of randomly right-censored data, and where some expert information on the tail index is available for the censored observations. This setting arises for instance naturally for liability insurance claims, where actuarial experts build reserves based on the specificity of each open claim, which can be used to improve the estimation based on the already available data points from closed claims. Through an entropy-perturbed likelihood we derive an explicit estimator and establish a close analogy with Bayesian methods. Embedded in an extreme value approach, asymptotic normality of the estimator is shown, and when the expert is clair-voyant, a simple combination formula can be deduced, bridging the classical statistical approach with the expert information. Following the aforementioned combination formula, a combination of quantile estimators can be naturally defined. In a simulation study, the estimator is shown to often outperform the Hill estimator for censored observations and recent Bayesian solutions, some of which require more information than usually available. Finally we perform a case study on a motor third-party liability insurance claim dataset, where Hill-type and quantile plots incorporate ultimate values into the estimation procedure in an intuitive manner.
\end{abstract}
\maketitle

\textbf{Keywords:} penalized likelihood, extreme value statistics, regular variation, survival analysis.

\section{Introduction}	

In applied statistics, one is often faced with the need to combine different types of information to produce a single decision. For instance, in credibility theory, the weights that link a relevant but small dataset with a big but not-so-relevant dataset are looked for, or similarly in bioinformatics one may use a dataset containing different cell types in the estimation of a cell, and scale down their importance in various ways.

The present paper is motivated by a specific problem in liability insurance. In that line of business, claim size data usually have a high percentage of censored observations, as policies take years, or even decades, to be finally settled. Due to the limited number of claims, one still would like  to take into account available information about the open claims in the estimation of claim size distributions (see e.g.\ \cite{abt}). On the one hand, experts typically project the final amount of open claims, i.e.\ they propose \textit{incurred values}, or also \textit{ultimates} based on covariate information or other (objective or subjective) considerations which are not in the payment dataset that arrives at a statistician's table. On the other hand, statisticians have standard ways of dealing with censored observations, for instance  the Peaks over Threshold method when one is interested in extremes, as well as the Hill estimator for heavy and Pareto-like tails. This research has started in \cite{beirlcens2007} and \cite{einmahl2008statistics} and has received more attention recently, see e.g.  \cite{worms2014new}, \cite{ameraoui2016bayesian}, \cite{beirlant2018penalized}.  However, in that line of extreme value methods expert information has not been incorporated. In \cite{abt}, incurred values were used to derive upper bounds for the open claims and survival analysis methods for interval censored data were implemented. See also \cite{lesaffre} for frequentist and Bayesian analysis of interval censored data.\\
One often faces the question of whether to conduct the analysis from the right-censored observations point of view, or whether to use the imputed ultimate (expert) values into the dataset and treat it as a fully-observed dataset. The latter is typically an easy (and cheap) solution. Figure \ref{description0} illustrates a possible situation of available data for motor third-party liability (MTPL) insurance claims of a direct insurance company operating in the EU, cf.\ Section \ref{MTPL_sect} for more details, where this data set will be studied. In what follows we are interested in developing a procedure that combines both approaches, without making any assumptions on the quality or method used to obtain the expert information. 

\begin{figure}[]
\centering
\includegraphics[width=10cm,trim=0.5cm 0.5cm 0.5cm 0.5cm,clip]{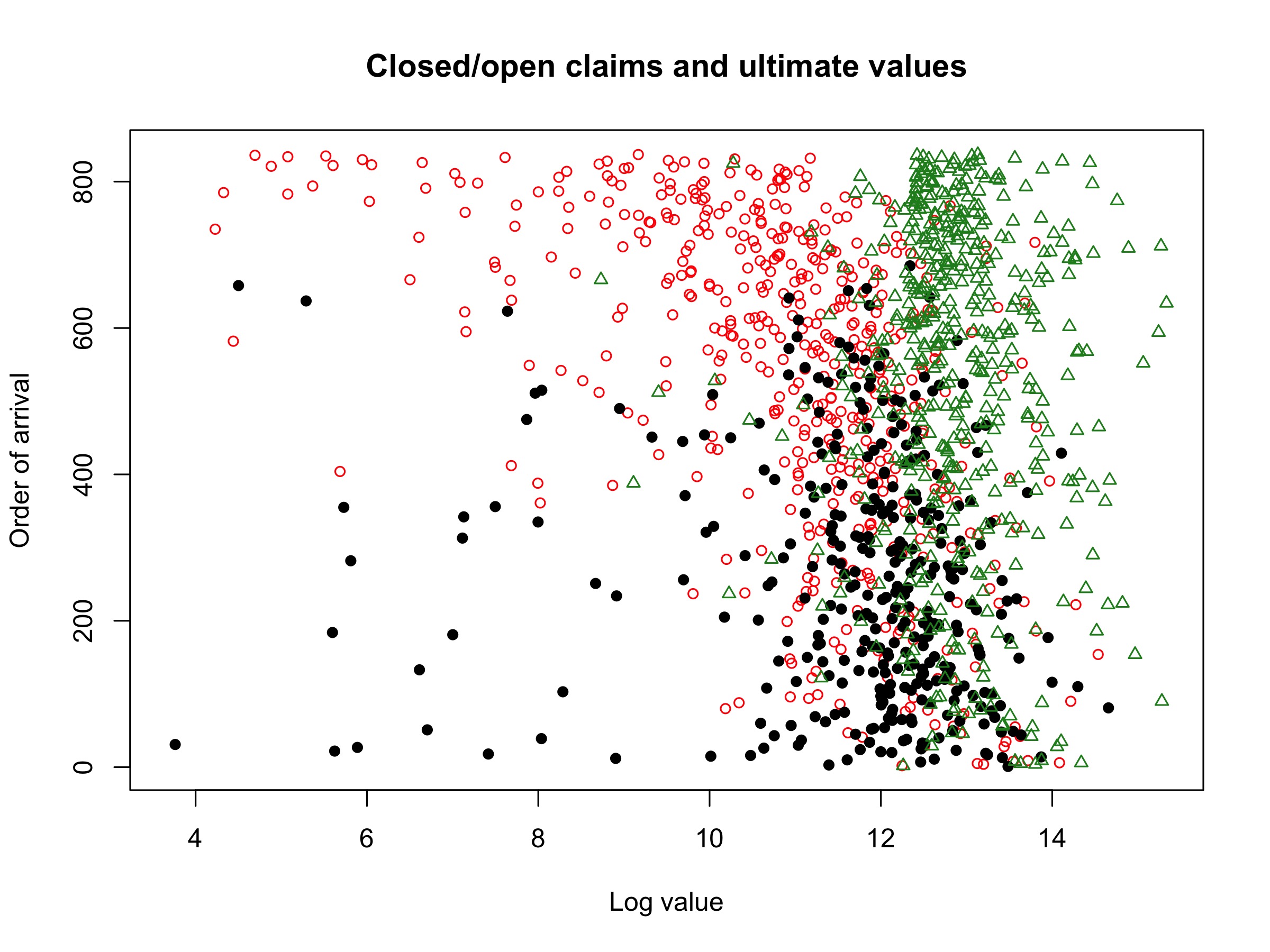}
\caption{Motor third-party liability insurance: log-claims in vertical order of arrival, showing the paid amount for both open (red, circle) and closed (black, dot) claims, as well as ultimate values for the open claims (green, triangle).} 
\label{description0}
\end{figure}

To that end, we assume that for each censored observation (open claim), we have a tail parameter $\beta_i$ which reflects the belief of the expert on the heaviness of the tail of this particular (unsettled)  observation. The typical situation may be that all the $\beta_i$ are equal or that there is an upper and a lower bound for all of them. However, we develop the theory for the general case, and we embed these indices into a statistical framework, where a single tail parameter is estimated for the entire dataset. At a philosophical level, the proposal of different $\beta_i$ is not an ill-posed problem,  it rather shows a prior variability of a belief on the tail index. 

The difference with the Bayesian paradigm is the fact that we only make the assumption for censored observations, such that increasing sample sizes but constant censored proportions will keep the importance of the expert guesses constant, with respect to the rest of the data. In mathematical terms we will see that our approach will have a Bayesian interpretation when the prior distribution of the parameter depends on the sample through the censoring indicators.

We propose a perturbation of the likelihood via an exponential factor and use the relative entropy between two densities as a dissimilarity measure. The resulting maximum perturbed-likelihood estimate has an explicit formula which resembles the Hill estimator adapted for censoring (cf. \cite{Hill}; \cite{beirlcens2007}), degenerates into it as the perturbation becomes small, and converges to the mean of the expert tail indices as the perturbation becomes large. Thus, in a similar way as in the prior specification for Bayesian estimation, if  experts have additional information on the quality of their belief, the perturbation parameter can be tuned accordingly. However, we propose a method which does not assume such additional prior knowledge, apart from the original expert knowledge.

Penalization is a prevalent idea which has gained popularity in the age of cheap computational power. The idea behind it is to impose beliefs on the statistical estimation which can yield a better estimation or an estimator which has more acceptable properties for the application at hand. It can help to impute a control or perturbation parameter which in turn helps to tailor estimators towards a certain degree of convenience. For instance, in a different statistical setting, the Lasso or Ridge regression imposes the belief or need of a data scientist to reduce the number of covariates included in a covariate-response analysis. In some cases this procedure helps to remove nuisance covariates, but in others it might be too aggressive and exclude truly informative variables. This bottleneck is specific to each application -- and even to each dataset -- which suggests that a fully automated procedure is not recommended. In the same vein, the effectiveness of the proposed method depends on the quality of the provided expert tail information, and this is something which is not always available or quantifiable. In any case, it is recommended that the tail inference is done side-by-side with the experts.


We derive the asymptotic properties of the perturbed-likelihood estimator, and although the asymptotic mean square error (AMSE) is available, the parameter which minimizes it depends on subjective considerations. One such consideration is the strength of the  belief in the expert information provided: if the belief is certain, the penalization parameter goes to infinity, which means that the data should be ignored and the expert information should be used instead, based on the AMSE criterion. To avoid assumptions which are not realistically available in practice, we instead suggest selecting the penalizing parameter in a convenient way, as the one which reduces the perturbed estimator to a simple sum of identifiable components. When the expert information is precise (degenerate at the true parameter) it holds that the penalization weight equal to 1 always leads to a lower AMSE, the latter also having a formula with a pleasant interpretation. When substituting this penalization parameter into the original formula, a very simple interpretation of the (inverse of the) estimator is available (Theorem \ref{lambdaone}): it is the combination of the Hill estimator and the expert information, the weights being the proportion of non-censored and censored data-points, respectively. Such a simple combination estimator is shown for a variety of common heavy-tailed distributions and for a variety of parameters to perform very well alongside competing methods, which need information to tune their own parameters.

The remainder of the paper is structured as follows. In Section \ref{deriv}, for the exact Pareto distribution, we introduce the notation and perturbation that we will deal with, as well as deriving simple expressions for the maximum perturbed-likelihood estimators and showing how they bridge theory and practice in a smooth manner. In Section \ref{Bayes_sect} we establish a close link with Bayesian statistics. In Section \ref{evt} we extend the methodology to the case where the data only exhibit Pareto-type behaviour in the tail, and we derive the asymptotic distributional properties of the perturbed-likelihood estimator and unveil a simple combination formula. In this more general heavy-tailed case, we naturally deal with estimators that use only a fraction of upper order statistics, and introduce as benchmarks some recent Bayesian estimators that have been proposed for censored datasets.  In Section \ref{MTPL_sect} we perform a simulation study and a real-life motor third-party liability insurance application. The latter data has been studied in the literature from both the expert information and the censored dataset viewpoints, but not yet in a joint manner, as we do here. We conclude in Section \ref{conclusion_sect}.

\section{Derivation and properties}\label{deriv}

Consider the estimation from a censored sample following an exact Pareto distribution. That is, we observe the randomly censored data-points and the binary indicator censoring variables:
\begin{align*}
(Z_1,\,\ee_1),\: (Z_2,\,\ee_2),\dots,(Z_n,\,\ee_n).
\end{align*}
Contrary to classical survival analysis, the $Z_i$ here correspond to payment sizes rather than times. The density and tail of the non-censored underlying data (which is not observed) are given by
\begin{align}\label{exactpareto}
f_\alpha(x)=\frac{\alpha x_0^\alpha}{x^{\alpha+1}},\quad \overline F_\alpha(x)=\frac{ x_0^\alpha}{x^{\alpha}}, \quad x\ge x_0>0,
\end{align}
with unknown tail parameter $\alpha$ and known scale parameter $x_0$. The latter assumption poses no restriction, since we are interested only in the estimation of the tail index, and as we will see, the Hill-type estimators based on upper order statistics that will be considered depend only on the log spacings of the data, which are independent of the scale parameter.

Additionally, and in contrast to classical survival analysis, we assume that we are given for each (right-)censored data-point an expert information of the possible tail parameter, i.e. we have   knowledge of $\beta_i>0$ for $i=1,\dots, n$. This can arise, for instance, when  the data are collected from different sources, or the realization of a data-point showing some pattern due to a particular settlement history, or some covariate information that can not be included in a more direct way. However, it is believed that all data points eventually come from one underlying distribution (or at least one aims for such a modelling description). We are also primarily interested in the case where all the $\beta_i$ are the same, since more often than not, expert information can come in this format.

When ignoring the information from the data, natural estimates of $\alpha$ are given by the weighted arithmetic and harmonic mean
\begin{align}\label{expavg}
\hat\alpha_{\text{am}}=\frac{\sum_{i=1}^n (1-\ee_i)\beta_i}{\sum_{i=1}^n (1-\ee_i)},\quad \hat\alpha_{\text{hm}}=\frac{\sum_{i=1}^n (1-\ee_i)}{\sum_{i=1}^n (1-\ee_i)/\beta_i},
\end{align}
respectively, where $\ee_i=0$ if $Z_i$ is right-censored, and $\ee_i=1$ otherwise.
On the other hand, in the context of survival analysis it is a standard approach to maximize the following likelihood based purely on the data
\begin{align*}
\mathcal{L}(\alpha;z)=\prod_{i=1}^n f_{\alpha}(z_i)^{\ee_i} \overline F_\alpha(z_i)^{1-\ee_i}=\prod_{i=1}^n \left(\frac{\alpha x_0^\alpha}{z_i^{\alpha+1}}\right)^{\ee_i} \left(\frac{ x_0^\alpha}{z_i^{\alpha}}\right)^{1-\ee_i}.
\end{align*}
 The maximum likelihood estimator is then given by
\begin{align}\label{hill}
\hat\alpha^{MLE}=\frac{\sum_{i=1}^n\ee_i}{\sum_{i=1}^n\log\left(\frac{Z_i}{x_0}\right)},
\end{align}
which is an adaptation (see \cite{beirlcens2007}) to the censoring case of the famous Hill  estimator (cf. \cite{Hill}) from extreme value theory obtained by Peaks-over-Threshold modelling; see \cite{embrechts2013modelling} or \cite{BGST2004} for a broader treatment on Pareto-type tail estimation, see also Section 4.

The two aforementioned approaches to the estimation of the tail parameter are in practice separated. That is, an expert will take only one of the two approaches, based on factors such as reliability of the expert information or data availability.  This is an especially difficult decision when there is a high percentage of censored and large observations, which present a key problem in the estimation in statistics in general, see for instance \cite{leung1997censoring}.

In the present paper, we introduce an estimator which bridges the previous estimators. We will proceed by the perturbation of the likelihood function, and consider a penalized likelihood:
\begin{align*}
\mathcal{L}^P(\alpha;z)&=\prod_{i=1}^n f_{\alpha}(z_i)^{\ee_i} \overline F_\alpha(z_i)^{1-\ee_i}e^{-(1-\ee_i)\lambda D(\alpha,\beta_i)}\\&=\prod_{i=1}^n \left(\frac{\alpha x_0^\alpha}{z_i^{\alpha+1}}\right)^{\ee_i} \left(\frac{ x_0^\alpha}{z_i^{\alpha}}\,e^{-\lambda D(\alpha,\beta_i)}\right)^{1-\ee_i},
\end{align*}
where the factor $e^{-\lambda D(\alpha,\beta_i)}$ penalizes the contribution of the censored observations according to some measure of dissimilarity between $f_\alpha$ and the Pareto distribution with parameter $\beta_i$, denoted by $D(\alpha,\beta_i)$, and the $\lambda\ge0$ models the strength of the penalization imposed by $D(\alpha,\beta_i)$. We propose to use the relative entropy as a dissimilarity measure:

\begin{align}\label{entropypenaliz}
D(\beta_i,\alpha)=\int_{x_0}^\infty\log\left(\frac{g_i(s)}{f_\alpha(s)}\right)g_i(s)\dd s=\frac{\alpha}{\beta_i} -1- \log\left(\frac{\alpha}{\beta_i}\right)\ge 0,
\end{align}
where $g_i$ is a Pareto density with tail index $\beta_i$ and scale parameter $x_0$. The associated log-likelihood is then given by

\begin{align}\label{exppen}
\log \left(\mathcal{L}^P(\alpha,z)\right)=\sum_{i=1}^n \ee_i\log\left(\frac{\alpha x_0^\alpha}{z_i^{\alpha+1}}\right)+\sum_{i=1}^n (1-\ee_i)\log\left(\frac{ x_0^\alpha}{z_i^{\alpha}}\right)-\sum_{i=1}^n\lambda (1-\ee_i)D(f_\alpha,g_i).
\end{align}
Equation \eqref{exppen} turns out to have an explicit minimizer when using $D$ from \eqref{entropypenaliz} (we omit the details), given by
\begin{align*}
\hat\alpha^{P}(\lambda)=\frac{\sum_{i=1}^n(\ee_i+\lambda(1-\ee_i))}{\sum_{i=1}^n(\log(Z_i/x_0)+\lambda(1-\ee_i)/\beta_i)}.
\end{align*}

Notice that if we flip the densities in the entropy penalization and consider instead

\begin{align*}
D(\alpha,\beta_i)=\int_{x_0}^\infty\log\left(\frac{f_\alpha(s)}{g_i(s)}\right)f_\alpha(s)\dd s=\frac{\beta_i}{\alpha} -1- \log\left(\frac{\beta_i}{\alpha}\right)\ge 0,
\end{align*}
the associated penalized likelihood has the explicit solution
\begin{align}\label{expenalpha}
&\hat\alpha^{I}(\lambda)=\\
&\frac{ \sum_{1}^n(\ee_i-\lambda(1-\ee_i))+\sqrt{ \left[\sum_{1}^n(\ee_i-\lambda(1-\ee_i)\right]^2+4\sum_{1}^n\log\left(\frac{Z_i}{x_0}\right)\cdot\sum_{1}^n\beta_i\lambda(1-\ee_i)}  }{2 \sum_{1}^n\log\left(\frac{Z_i}{x_0}\right)},\nonumber
\end{align}
which is less appealing, with more complicated asymptotic properties.

\begin{remark}\normalfont
The particular choice of entropy penalization is mathematical in nature, since the resulting explicit and simple form of the maximum likelihood estimator permits a deeper analysis than other choices. For instance, the significantly more complicated explicit estimators \eqref{expenalpha} or \eqref{gaussestim} lead to a much more involved analysis.
\end{remark}

\begin{remark}\normalfont
In general, with lack of any other type of information, giving equal weight to each censored observation is the most natural way to deal with them. If the expert has an idea of the importance of each data point and their corresponding tail indices $\beta_i$, the selection can be done on a single parameter $\lambda$ through
\begin{align*}
\lambda_i=\lambda \omega_i.
\end{align*}
Note that then
\begin{align}\label{l2}
\lim_{\lambda\to \infty} \hat\alpha^{P}(\lambda)=\frac{\sum_{i=1}^n(1-\ee_i)\omega_i}{\sum_{i=1}^n(1-\ee_i)\omega_i/\beta_i},
\end{align}
and

\begin{align}\label{l1}
\lim_{\lambda\to \infty} \hat\alpha^{I}(\lambda)=\frac{\sum_{i=1}^n(1-\ee_i)\omega_i\beta_i}{\sum_{i=1}^n(1-\ee_i)\omega_i},
\end{align}
i.e.\ the information brought by the data becomes irrelevant and we take a weighted average of the expert guesses. Taking uniform weights, i.e., giving equal importance to each censored observation, will result in \eqref{expavg}.

If no weights are naturally suggested one can always tackle the multi-dimensional selection problem on all $\lambda_i$. In this more general case we have that 
\begin{align}\label{l3}
\lim_{\lambda_i\to 0;\:i=1,\dots,n} \hat\alpha^{P}(\lambda_1,\dots,\lambda_n)=\frac{\sum_{i=1}^n\ee_i}{\sum_{i=1}^n\log\left(\frac{Z_i}{x_0}\right)}.
\end{align}
which can readily be deduced directly from \eqref{exppen}, since it is the classical non-penalized estimator. Similarly
\begin{align}\label{l4}
\lim_{\lambda_i\to 0;\:i=1,\dots,n} \hat\alpha^{I}(\lambda_1,\dots,\lambda_n)=\frac{\sum_{i=1}^n\ee_i}{\sum_{i=1}^n\log\left(\frac{Z_i}{x_0}\right)}.
\end{align}

  As a consequence of the limits \eqref{l2}, \eqref{l1}, \eqref{l3} and \eqref{l4}, we readily get  for $\lambda_1=\cdots=\lambda_n=:\lambda\ge 0$ that
  \begin{align*}
\lim_{\lambda\to \infty}\hat\alpha^{P}(\lambda)= \hat\alpha_1,\quad\lim_{\lambda\to \infty} \hat\alpha^{I}(\lambda)= \hat\alpha_2,\quad \lim_{\lambda\to 0}\hat\alpha^{P} =\lim_{\lambda\to 0}\hat\alpha^{I}(\lambda)= \hat\alpha^{MLE},
\end{align*}
which confirms that the estimator bridges the estimation of $\alpha$ and the proposal of the $\beta_i$, and that the parameter $\lambda$ reflects in some sense the strength of the belief on the expert information. The next section will touch upon this interpretation in a more precise manner.
\end{remark}

\section{Penalization seen as a Bayesian prior}\label{Bayes_sect}
We will use a single $\lambda$ value in practice, but here we assume the most general setting where the $\lambda_i$ could be different, at no complexity cost. The penalized likelihood that gives rise to $\hat\alpha^{P}$ is given by

\begin{align*}
\mathcal{L}^P(\alpha;z)&=\prod_{i=1}^n \left(\frac{\alpha x_0^\alpha}{z_i^{\alpha+1}}\right)^{\ee_i} \left(\frac{ x_0^\alpha}{z_i^{\alpha}}\right)^{1-\ee_i}e^{-\lambda_i(1-\ee_i) (\alpha/\beta_i-1-\log(\alpha/\beta_i))}\\
&=\left[\prod_{i=1}^n \left(\frac{\alpha x_0^\alpha}{z_i^{\alpha+1}}\right)^{\ee_i} \left(\frac{ x_0^\alpha}{z_i^{\alpha}}\right)^{1-\ee_i}\right]\cdot \left[\alpha^{\sum_{i=1}^n\lambda_i(1-\ee_i)}e^{-\alpha \sum_{i=1}^n\lambda_i(1-\ee_i)/\beta_i}\right]\\
&\quad\times\left[\prod_{i=1}^n\beta_i^{-\lambda_i(1-\ee_i)}e^{\lambda_i(1-\ee_i)}\right]\\
&=\left[\alpha^{\sum_{i=1}^n(\ee_i+\lambda_i(1-\ee_i))}e^{-\alpha\sum_{i=1}^n(\lambda_i(1-\ee_i)/\beta_i+\log(z_i/x_0))}\right]\\
&\quad \times\left[\prod_{i=1}^n\beta_i^{-\lambda_i(1-\ee_i)}e^{\lambda_i(1-\ee_i)}z_i^{-\ee_i}\right].
\end{align*}

\noindent Note that the second factor after the last equality sign does not depend on $\alpha$, and the first one is proportional to a gamma density. We thus recognize that the penalized maximum likelihood estimator can be seen as the posterior mode arising from a Pareto likelihood and the conjugate gamma prior with hyper-parameters
\begin{align*}
\alpha_0=\sum_{i=1}^n\lambda_i(1-\ee_i)+1,\quad \beta_0=\sum_{i=1}^n\lambda_i(1-\ee_i)/\beta_i,
\end{align*}
and corresponding posterior parameters
\begin{align*}
\alpha^\ast=\sum_{i=1}^n(\ee_i+\lambda_i(1-\ee_i))+1,\quad \beta^\ast=\sum_{i=1}^n(\lambda_i(1-\ee_i)/\beta_i+\log(z_i/x_0)).
\end{align*}

The hyper-parameters of the prior, however, do depend on the sample, namely on the censoring indicators $\ee_i$, so we are not in the classical Bayesian setting. Nonetheless, we will continue to call it a prior, for simplicity.

In this context we also have the following interpretation of the effects of the selection of the $\lambda_i$. The mode of the prior distribution is given by
\begin{align}\label{priormodevar}
\frac{\sum_{j=1}^n\lambda_j(1-\ee_j)}{\sum_{i=1}^n\lambda_i(1-\ee_i)/\beta_i}=\left(\sum_{i=1}^n\frac{\lambda_i(1-\ee_i)}{\sum_{j=1}^n\lambda_j(1-\ee_j)}\beta_i^{-1}\right)^{-1},
\end{align}
and one sees that the proportions of the $\lambda_i$ give the weights which will determine this mode. However, we can multiplicatively scale these $\lambda_i$ and the mode will remain unchanged. The magnitude of the $\lambda_i$, in contrast, does play a role for the variance of the prior:
\begin{align}
\frac{\sum_{i=1}^n\lambda_i(1-\ee_i)+1}{\left(\sum_{i=1}^n\lambda_i(1-\ee_i)/\beta_i\right)^2},
\end{align}
since the larger the $\lambda_i$, the smaller the prior variance. Thus, a single estimate as an expert information will trump the ability to effectively determine the magnitude of the penalization parameter. This is a problem which is often encountered in Bayesian statistics, and a prior is often selected nonetheless, making frequentists doubtful of this philosophical leap of faith. 

Note that the gamma distribution has two parameters, and any two descriptive statistics (presently we used the mode and variance) which bijectively map to the mode and variance can be used to give alternate full explanations as to how the proportions $\lambda_i(1-\ee_i)/\sum_{j=1}^n\lambda_j(1-\ee_j)$ and the sizes of the $\lambda_i$ play a role in the modification of the prior distribution, and hence on the expert belief.

\begin{remark}\normalfont
If instead of using the penalization given by \eqref{entropypenaliz} we simply use squared penalization given by

\begin{align*}
D(\beta_i,\alpha)=\frac{(\alpha-\beta_i)^2}{2}\ge 0,
\end{align*}
then the maximum perturbed-likelihood estimate will again be explicit and given by
\begin{align}\label{gaussestim}
\hat\alpha^{Sq}=&\frac{\sum_{i=1}^n(\lambda_i(1-\ee_i)\beta_i-\log(Z_i/x_0))}{\sum_{i=1}^n\lambda_i(1-\ee_i)}\nonumber\\
&+\frac{\sqrt{ \left[\sum_{i=1}^n(\lambda_i(1-\ee_i)\beta_i-\log(Z_i/x_0)) \right]^2+4\sum_{i=1} ^n \lambda_i(1-\ee_i)\cdot \sum_{i=1}^n\ee_i}}{\sum_{i=1}^n\lambda_i(1-\ee_i)},
\end{align}
which  naturally leads to a Gaussian prior interpretation when the $\lambda_i$ are equal. This estimator also converges to the Hill estimator as $\lambda_i\to0$, $i=1,\dots, n$, but it can have numerical instabilities when the denominator becomes very small. 
\end{remark}

\section{Extreme Value Theory}\label{evt}

We now move on to a more general heavy-tail approach and consider the case of regularly varying distributions with tail of the form 
\begin{align*}
x^{-\alpha}\ell(x), \quad \alpha>0,
\end{align*}
where $\ell$ is a slowly varying function, i.e. ${\ell (vx) \over \ell (x)} \to 1$, as $x \to \infty$ for every $v>1$. We also assume now that censoring is done at random and the data is generated as the minimum of two independent random variables 
\begin{align*}
Z_i=\min\{X_i,L_i\},
\end{align*}
with regularly varying tails:
\begin{align*}
 &\Prob(X_i>u)=u^{-\alpha}\ell(u),\\
 &\Prob(L_i>u)=u^{-\alpha_2}\ell_2(u).
\end{align*}
It follows that
\begin{align}\label{randcens}
\Prob(Z_i>u)=u^{-\alpha_c}\ell_c(u), \quad \alpha_c=\alpha+\alpha_2,
\end{align}
and slowly varying function $\ell_c=\ell\, \ell_2$.
Here we confine ourselves to the so-called \textit{Hall class} (cf.\ \cite{hall82}). This popular second-order assumption in extreme value theory often makes asymptotic identities tractable:
\begin{eqnarray}\label{hallclass}
\Prob(X_i>u) &=& C_1 u^{-\alpha}\left( 1+D_1 u^{-\nu_1}(1+o(1))\right) \mbox{ for } u \to \infty, \nonumber \\
\Prob(L_i >u) &=& C_2 u^{-\alpha_2}\left( 1+D_2 u^{-\nu_2}(1+o(1))\right) \mbox{ for } u \to \infty,
\label{HW}
\end{eqnarray} 
where $\nu_1,\nu_2, C_1, C_2$ are positive constants and $D_1,D_2$ real constants. Then, with 
\begin{align*}
C=C_1C_2,\quad \nu_*=\mbox{min} (\nu_1,\nu_2)
\end{align*}
and
\begin{align*}
D_* =\begin{cases}
D_1 , &\nu_1 < \nu_2\\
D_2, & \nu_2 < \nu_1\\ 
D_1+D_2,& \nu_1 = \nu_2,
\end{cases}
\end{align*}
we have that 
\[
\Prob(Z_i>u) = Cu^{-\alpha_c}
\left(1+D_* u^{-\nu_*}(1+o(1))\right),
\]
that is, the censored dataset is again in the Hall class. 

Denote the quantile function of $Z$ by $Q$ and consider the tail quantile function $U(x)=Q(1-x^{-1})$, $x>1$. Then we have that
\[
U(x) = C^{1/\alpha_c}\left( 1+ {D_* \over \alpha_c}C^{-\nu_*/\alpha_c}x^{-\nu_*/\alpha_c}(1+o(1))\right).
\] 
The order statistics of the data will be denoted by
\[
Z^{(1)}\ge \cdots \ge Z^{(n)},
\]
with associated censoring indicators $\ee^{(i)}$ and expert information $\beta^{(i)}$. Given a high threshold $u>x_0$, the Hill estimator adapted for censoring is 
\begin{align}\label{hillestimdef}
\hat\alpha^H_u=\frac{\sum_{i=1}^n\ee_i1\{Z_i>u\}}{\sum_{i=1}^n\log\left(\frac{Z_i}u\right)1\{Z_i>u\}}.
\end{align}
Taking  $Z^{(k)}$ for some $1\le k\le n$, as a (random) threshold $u$, we obtain the alternative order statistics version
\begin{align}\label{mlealpha}
\hat\alpha^{MLE}_k=\frac{\sum_{i=1}^k\ee^{(i)}}{\sum_{i=1}^k\log\left(\frac{Z^{(i)}}{Z^{(k+1)}}\right)} = {\hat{p}_k\over H_k},
\end{align}
where 
$$\hat{p}_k = {1 \over k}\sum_{i=1}^k \ee^{(i)}$$ 
is the proportion of non-censored observations in the largest $k$ observations of $Z$, and
$$H_k = {1 \over k}\sum_{i=1}^k\log\left(\frac{Z^{(i)}}{Z^{(k+1)}}\right)$$ is the classical Hill estimator based on the largest $k$ observations. For details on these censored versions of the Hill estimator, we refer to \cite[Sec.2]{einmahl2008statistics}.

The asymptotic distribution of $H_k$ has been studied intensively in the literature under the above second-order assumptions (see for instance \cite[Ch.4]{BGST2004}): assuming
\begin{align}\label{todelta}
\sqrt{k}(k/n)^{\nu_*/\alpha_c} \to \delta\ge0,
\end{align}
as $k,n \to \infty$ with $k/n \to 0$, we have that
\begin{align}\label{y0}
\sqrt{k}\left( H_k -{1 \over \alpha_c}\right) \stackrel{d}{\to} 
Y_0 \sim \mathcal{N}\left(
-C^{-\nu_*/\alpha_c}D_*{\nu_*\delta\over \alpha_c (\alpha_c+ \nu_*)},
 \alpha_c^{-2}\right).
\end{align}
As discussed in \cite{einmahl2008statistics}, the asymptotic bias of $\hat{p}_k$ follows from the leading term  in ${1 \over k}\sum_{i=1}^k p(U(n/i)) -p$, where 
\[
p(z)= \mathbb{P} \left(e =1 |Z=z \right),
\]
and  $p$ denotes the  asymptotic probability of non-censoring 
\begin{align*}
p=\lim_{z \to \infty}p(z) =\frac{1/\alpha_2}{1/\alpha+1/\alpha_2}=\frac{\alpha}{\alpha+\alpha_2}.
\end{align*}
Under the Hall class \eqref{HW}, we have with the definition
\[
(D/\alpha)_* =\begin{cases}
D_1/\alpha , &\nu_1 < \nu_2\\
-D_2/\alpha_2, & \nu_2 < \nu_1\\ 
D_1/\alpha-D_2/\alpha_2,& \nu_1 = \nu_2,
\end{cases}
\]
that as $x \to \infty$ 
\begin{equation}
p(U(x)) -p = p(1-p) (D/\alpha)_* \nu_* C^{-\nu_*/\alpha_c}x^{-\nu_*/\alpha_c}
(1+o(1)).
\label{biaspU}
\end{equation}
From this, assuming that $\sqrt{k}(k/n)^{\nu_*/\alpha_c} \to \delta$ as $k,n \to \infty$ with $k/n \to 0$, one gets
\[
\sqrt{k}\left( \hat{p}_k -p \right) \stackrel{d}{\to} 
\mathcal{N}\left(
p(1-p)C^{-\nu_*/\alpha_c} (D/\alpha)_* {\alpha_c \nu_*\delta \over \alpha_c + \nu_*},
 p(1-p)\right).
\]
In \cite{einmahl2008statistics} it was also derived that asymptotically $H_k$ and $\hat{p}_k$ are independent, so that under the condition \eqref{todelta} as $k,n \to \infty$ with $k/n \to 0$,
\begin{align}\label{noncomb}
\sqrt{k}\left(\frac{1}{\hat\alpha^{MLE}_k}-\frac{1}{\alpha}\right)
&\stackrel{d}{\to} 
\mathcal{N}\left(-{\delta\nu_* \over \alpha_c +\nu_*} C^{-\nu_*/\alpha_c}[D_* (\alpha_c^{-1}+\alpha^{-1})+ {\alpha_2 \over \alpha}(D/\alpha)_*], \frac{1}{p\alpha^2} \right).
\end{align}

\vspace{0.3cm}

In the same manner we can define a version of $\hat\alpha^{P}$ which perturbs at censored data-points and which considers only large claims. We consider as before that $\lambda_i=\lambda$, and, in analogy to the exact Pareto setting, define the two estimators
\begin{align*}
\hat\alpha^{P}_u=\frac{\sum_{i=1}^n(\ee_i+\lambda(1-\ee_i))1\{Z_i>u\}}{\sum_{i=1}^n(\log(Z_i/u)+\lambda(1-\ee_i)/\beta_i)1\{Z_i>u\}},
\end{align*}
and the order statistics version
\begin{align*}
\hat\alpha^{P}_k
&=\frac{\sum_{i=1}^k(\ee^{(i)}+\lambda (1-\ee^{(i)}))}{\sum_{i=1}^k\left(\log\left(\frac{Z^{(i)}}{Z^{(k)}}\right)+\lambda (1-\ee^{(i)})/\beta^{(i)}\right)}.
\end{align*}


\begin{theorem}\label{biasvar}
Assume \eqref{HW}.
Set $\lambda_i=\lambda\ge 0$, $\beta_i=\beta>0$. As $\sqrt{k}(k/n)^{\nu_*/\alpha_c} \to \delta$, as $k,n \to \infty$ with $k/n \to 0$,  
\[
\sqrt{k} \left(
{1\over \hat\alpha^{P}_k} - \frac{\lambda\alpha_2/\beta+1}{\lambda  \alpha_2+\alpha}\right)
\] is asymptotically normal with asymptotic mean

\begin{eqnarray}
\mathcal{M}&=&- 
 {\delta \nu_* C^{-\nu_*/\alpha_c}\over 1-r_1}
 \left(\frac{D_*/\alpha_c  +
  \lambda p(1-p) (D/\alpha)_* \alpha_c/\beta}{\nu_* +\alpha_c}\right. \nonumber\\
  && \left. \hspace{2.8cm}+
  {\lambda r_2 + \alpha_c^{-1} \over 1-r_1}p(1-p) (D/\alpha)_*\left( {\alpha_c \over \alpha_c+\nu_*}\right)
 \right) \nonumber
\end{eqnarray}

\noindent and variance
\begin{align}\label{variance}
\mathcal{V}=\frac{1}{\alpha_c^2(1-r_1)^2}+\frac{1}{(1-r_1)^4}\left(\frac{\lambda}{\beta(1-\lambda)}+\frac{1}{\alpha_c}\right)^2(1-\lambda)^2p(1-p),
\end{align}
where $ r_1=(1-p)(1-\lambda)$ and  $r_2=(1-p)/\beta$. The asymptotic bias of $1/\hat{\alpha}^P_k$ equals 
\begin{align}\label{bias}
\mathcal{B}=\frac{\lambda \alpha_2/\beta+1}{\lambda  \alpha_2+\alpha}-{1 \over \alpha}+O\left((k/n)^{\nu_*/\alpha_c}\right)
\end{align}
as $k,n \to \infty$ and $k/n \to 0$.
\end{theorem}
\begin{proof}
See Appendix A.
\end{proof}

\begin{remark}\normalfont
Notice that estimates of $\alpha_2$ or $\alpha_c$ are available using basic survival analysis techniques, cf. \eqref{randcens}. Consequently, we can use the plug-in method for the estimation of any of the above formulas that involve these quantities.
\end{remark}
\begin{remark}\normalfont
As a sanity check, observe that in Theorem \ref{biasvar}, whenever $\beta=\alpha$ and $\delta=0$, the bias vanishes.
\end{remark}

\vspace{0.2cm}\noindent
In the same spirit, even more can be said:
\begin{corollary}\label{lambdaone}
(Combination) Assume the conditions of Theorem \ref{biasvar}, and further $\delta=0$, with $\beta=\alpha$. Then the estimator $\hat\alpha^{P}_k$ with $\lambda=1$ is unbiased and can be written as
\begin{align*}
\hat\alpha^{P}_k=\left(\frac{\sum_{i=1}^k \ee^{(i)}}{k}\cdot \frac{\sum_{i=1}^k\log(Z^{(i)}/Z^{(k+1)})}{\sum_{i=1}^k \ee^{(i)}}+\frac{\sum_{i=1}^k(1-\ee^{(i)})}{k}\cdot \beta^{-1}\right)^{-1}.
\end{align*}
In words, $1/\hat\alpha^{P}_k$ is the weighted average of the MLE estimator and the expert information, the weights being the proportion of non-censored (and censored, respectively) observations above the threshold $T^{(k)}$.
Moreover, its inverse has asymptotic variance (and hence mean square error) given by
\begin{align*}
\operatorname{Var}(1/\hat\alpha_k^P)=\frac{1}{k p(\alpha+\alpha_2)^2},
\end{align*}
which, when compared to \eqref{noncomb}, is seen to enhance estimation.
\end{corollary}
The proof of Corollary \ref{lambdaone} is immediate.

\begin{remark}\normalfont

Observe that in a Bayesian setting, whenever we are aware that a parameter lies within an interval, a natural estimator is constructed as follows. We set a uniform prior on  $[b_1,b_2]$ and together with the Pareto likelihood we use the posterior mean as an estimate. Such a mean is given by

\begin{align*}
&\frac{\int_{b_1}^{b_2}\alpha^{1+\sum_{i=1}^n\ee_i}e^{-\alpha\sum_{i=1}^n\log\left(\frac{Z_1}{x_0}\right)} \dd\alpha}{\int_{b_1}^{b_2}\alpha^{\sum_{i=1}^n\ee_i}e^{-\alpha\sum_{i=1}^n\log\left(\frac{Z_1}{x_0}\right)} \dd\alpha}\\
=&\frac{\sum_{i=1}^n\ee_i+1}{\sum_{i=1}^n\log \left(\frac{Z_1}{x_0}\right)}\nonumber\\
 &\times\left[\frac{\gamma(\sum_{i=1}^n\ee_i+2,b_2\sum_{i=1}^n\log\left(\frac{Z_1}{x_0}\right))-\gamma(\sum_{i=1}^n\ee_i+2,b_1\sum_{i=1}^n\log\left(\frac{Z_1}{x_0}\right))}{\gamma(\sum_{i=1}^n\ee_i+1,b_2\sum_{i=1}^n\log\left(\frac{Z_1}{x_0}\right))-\gamma(\sum_{i=1}^n\ee_i+1,b_1\sum_{i=1}^n\log\left(\frac{Z_1}{x_0}\right))}\right]\nonumber,
\end{align*}
where
\begin{align*}
\gamma(u,v)=\frac{\int_0^vt^{u-1}e^{-t}\dd t}{\Gamma(u)}
\end{align*}
is the (normalized) lower incomplete gamma function. One can go one step further and define the order statistics version of the above estimator. However, despite being theoretically neat, the latter estimator is numerically unstable for both large ($k>100$) and small ($k<5$) number of upper order statistics, and hence we will not pursue it in the simulation section.


\end{remark}

\begin{remark}\normalfont
In \cite{ameraoui2016bayesian}, several Bayesian approaches for heavy tail estimation were considered (see also \cite{beirlant2018penalized}) under the random censoring assumption. We will use two of them as a benchmark. The first one arises from the posterior mean of a Pareto likelihood and the conjugate Gamma($a,b$) prior:
\begin{align}\label{bayesgamma}
\hat\alpha^{BG}=\frac{a+\sum_{i=1}^k \ee^{(i)}}{b+\sum_{i=1}^k\log(Z^{(i)}/Z^{(k+1)})}.
\end{align}
In the presence of a single expert estimate $\beta$ of the tail index, the prior parameters can be tuned by moment matching, where the variance will need to be imposed subjectively. That is, solve $\beta=a/b$ and $\sigma^2=a/b^2$ for an expert opinion on $\sigma^2$.

The second one arises from the maximal data information prior, and leads to the estimator
\begin{align}\label{bayesmaximal}
\hat\alpha^{BM}=\frac{1+\sum_{i=1}^k \ee^{(i)}+\sqrt{(1+\sum_{i=1}^k \ee^{(i)})^2+4\sum_{i=1}^k\log(Z^{(i)}/Z^{(k+1)})}}{2 \sum_{i=1}^k\log(Z^{(i)}/Z^{(k+1)})}.
\end{align}
Notice that the latter does not admit tuning the prior to additional data.

\end{remark}

\subsection*{Quantile estimation}
With the last result at hand it is natural to propose a quantile estimator based on the approach taken in \cite{weissman1978estimation}. Recall that we denote the quantile function of a regularly varying tail by $Q(p)$. Exploiting the fact that
\begin{align}\label{weismannasympt}
\frac{Q(1-p)}{Q(1-k/n)}\sim\frac{p^{-1/\alpha}}{(k/n)^{-1/\alpha}}=\left(\frac{k}{np}\right)^{1/\alpha}, \quad p\downarrow 0, \:k/n \to 0,\: np=o(k),
\end{align}
the Weissman estimator based on $k$ order statistics  (and without expert information) arises naturally as
\begin{align}
\hat Q^{MLE}_k(1-p)=\hat Q^{KM}(1-k/n)\cdot\left(\frac{k}{np}\right)^{1/\hat\alpha^{MLE}_k},
\end{align}
where $\hat Q^{KM}$ is the quantile function derived from the Kaplan-Meier estimator $$\widehat {S}(z)=\prod \limits _{i:\ Z_{i}\leq z}\left(1-{\frac {d_{i}}{n_{i}}}\right),$$ for the survival curve of the censored dataset in question, $(Z_i,\ee_i)$, $i=1,\dots,n$, where the $Z_i$ are payments (which would correspond to times in classical survival analysis terminology). Here, $d_i$ is the number of closed claims of a given size $z$, and $n_i$ is the number of payments, which irrespectively of censoring, are above $z$. In the case of no censoring this reduces to the empirical quantiles of the dataset, since the Kaplan-Meier curve is then just the empirical distribution function. Similarly, in the case of pure expert information an estimator can be proposed as
\begin{align}
\hat Q^{EX}_k(1-p)=\hat Q^{EX}(1-k/n)\cdot\left(\frac{k}{np}\right)^{1/\beta},
\end{align}
where $\hat Q^{EX}$ is either an expert-given cumulative distribution function, or in absence of it, simply the Kaplan-Meier quantiles. To combine these two results, Theorem \ref{lambdaone} leads the way. For the choice $\lambda=1$, we see that the Pareto part of the tail splits for the perturbed estimator according to
\begin{align*}
\left(\frac{k}{np}\right)^{1/\hat\alpha^P}=\left(\frac{k}{np}\right)^{\hat p_k/\hat\alpha^{MLE}_k}\cdot\left(\frac{k}{np}\right)^{(1-\hat p_k)/\beta}, 
\end{align*}
where 
\begin{align*}
\hat p_k=\frac1k \sum_{i=1}^k \ee^{(i)},
\end{align*}
and hence the following estimator is proposed for the overall tail
\begin{align}\label{perturbed_quantile}
\hat Q^{P}_k(1-p)&=\left[\hat Q^{KM}(1-k/n)\cdot\left(\frac{k}{np}\right)^{1/\hat\alpha^{MLE}_k}\right]^{\hat p_k}\cdot \left[\hat Q^{EX}(1-k/n)\cdot\left(\frac{k}{np}\right)^{1/\beta}\right]^{1-\hat p_k}\\
&=\hat Q^{P}(1-k/n)\cdot\left(\frac{k}{np}\right)^{1/\hat\alpha^P_k},
\end{align}
where 
\begin{align*}
\hat Q^{P}(1-k/n)=(\hat Q^{KM}(1-k/n))^{\hat p_k}(\hat Q^{EX}(1-k/n))^{1-\hat p_k}.
\end{align*}
Observe that in the absence of expert information for the quantile function, we merely have
\begin{align*}
\hat Q^{P}_k(1-p)=\hat Q^{KM}(1-k/n)\cdot\left(\frac{k}{np}\right)^{1/\hat\alpha^P_k}.
\end{align*}

\section{Simulation Study and MTPL Insurance}\label{MTPL_sect}

We perform in this section a simulation study and apply our method to a motor third party liability insurance dataset (cf. \cite[Sec.1.3.1]{abt}). In order to make our results comparable with existing studies and existing analysis of the aforementioned dataset, we will consider estimation of
$$\xi=\frac 1\alpha,$$
and thus we will make use of the estimators
\begin{align}\label{xiestims}
\hat\xi_k^{MLE}=\frac{1}{\hat\alpha^{MLE}_k},\:\: \hat\xi_k^P=\frac{1}{\hat\alpha^{P}_k},\:\: \hat\xi_k^{BG}=\frac{1}{\hat\alpha^{BG}_k},\:\: \hat\xi_k^{BM}=\frac{1}{\hat\alpha^{BM}_k}.
\end{align}

\subsection{Simulation Study}
We consider three heavy tails belonging to the Hall class \eqref{hallclass}, and compare $\hat\xi_k$ and the quantile estimator
\begin{align*}
\hat Q^{P}_k(1-p)=\hat Q^{KM}(1-k/n)\cdot\left(\frac{k}{np}\right)^{\hat\xi_k},
\end{align*}
where $\hat\xi_k$ is one of the four estimators in \eqref{xiestims}, and for $p=0.005$. For any tail estimator $\hat \xi_k$, we generically refer to $\hat Q^{P}_k$ as the corresponding Weismann estimator, since it was derived using the general principle of equation \eqref{weismannasympt}.

\noindent Concretely, we simulate two independent i.i.d.\ samples of size $n=200$, corresponding to the variables $X_i$ and $L_i$, $i=1,,\dots,n$, in \eqref{hallclass}. We repeat the procedure $N_{{sim}}=1000$ times. The following three distributions are employed, with two sub-cases for each distribution, for varying parameters:

\begin{itemize}
\item The exact Pareto distribution, defined in \eqref{exactpareto}, for $\xi=1,\,1/2$.\\

\item The Burr distribution, with tail given by
\begin{align*}
\overline F(x)=\left(\frac{\eta}{\eta+x^\tau}\right)^\lambda, \;x>0,\quad \eta,\tau,\lambda>0,
\end{align*}
We consider $\eta=1$, $\lambda=2$, $\tau=1/2$; $\eta=2$, $\lambda=1$, $\tau=2$; and $\eta=2$. Notice that $\xi=1/(\lambda\tau)$\\

\item The Fr\'{e}chet distribution with tail 
\begin{align*}
\overline F(x)=1-\exp(-x^{-\alpha}), \quad \alpha>0.
\end{align*}
We consider $\xi=1,\: 1/2$.\\
\end{itemize}

\noindent For the expert information we draw a single random number from a Gaussian distribution centered at the true $\xi$ and with standard deviation $0.2$, and define that value as $1/\beta$. Then, by moment matching, using a variance of $0.04$, we obtain the parameters $a,b$ needed for $\hat\xi^{BG}_k$. Notice that we input the true value of the variance, and hence we are giving additional information to the Bayesian setting, opposed to $\hat\xi_k^P$, where we make no such assumptions and we use the  combination with $\lambda=1$. Additional studies (which we omit here) show that if the Bayesian variance is not correctly specified (for instance, set at $1$ or $0.5$), the Bayesian solution behaves almost identically to the censored Hill estimator $\hat\xi_k^{MLE}$. 

Also notice the misspecification of the Gamma prior in the derivation of  $\hat\xi^{BG}_k$ with respect to the Gaussian distribution from which the expert information is actually simulated. Using a Gaussian prior would not only make explicit posterior formulas not available (and hence the need to resort to MCMC sampling methods as Gibbs sampling), but would add more information than what we have assumed is available throughout the paper (we have not even assumed knowledge of the variance).

We then plot the empirical bias and MSE of each resulting estimator as a function of $k$ (comparing the estimates with respect to the true value).  We write expressions such as Burr($\xi=1$) to indicate that the parameters of the distribution are not the focus, but rather the resulting tail index from the Hall class (which is a function of the parameters). The results are given in Figures \ref{sspareto}, \ref{ssburr} and \ref{ssfrechet}. We observe how the fact that the perturbation will affect estimation based on the proportion of censored observations as opposed to the total amount of data-points performs well for $k>10$. As a result, a substantial amount of bias and MSE is removed. This is especially the case for the heavy tail case $\xi=1$. For the lighter tail $\xi=0.5$, the perturbed estimator has either the best or second best performance bias-wise, and its only major drawback is the MSE for the exact Pareto case for $k>50$, where the Bayesian gamma solution performs even worse. When considering quantiles, the perturbed estimator behaves better than the Hill estimator and on par with the other two benchmarks for the heavy tail exact Pareto case. In the lighter tail case it performs the worst for large order statistics, recovering and behaving as in the previous case for $k<60$. In all other non-exact Pareto tail cases, the perturbation was superior to all methods. 

Notice that one assumption made in this study was that the expert guess was centered and with relatively good quality (mean $\xi$ and standard deviation $0.2$). If the latter conditions are changed, it is easy to construct a simulation study where both the perturbed and the Bayesian gamma solution perform much worse. Consequently, the findings of this simulation study suggest that insurers that are very confident in their expert opinions might benefit from using the combination estimator $\hat\xi_k^P$ with $\lambda=1$. 

\begin{remark}\normalfont
For the adaptive selection of $\lambda$, the procedure of cross-validation may naturally come to mind. However, the latter is based on averages rather than extremes. For instance, in a 10-fold cross-validation, the 9 parts of the data that do not contain the maximum will tend to prefer lighter tail indices, and only the one part with the maximum will suggest a heavy tail index, so that the overall index will be 
underestimated. Correspondingly, cross-validation is not a method of choice in this context.
\end{remark}

\begin{figure}[]
\centering
\includegraphics[width=12cm,trim=0.5cm 0.5cm 0.5cm 0.5cm,clip]{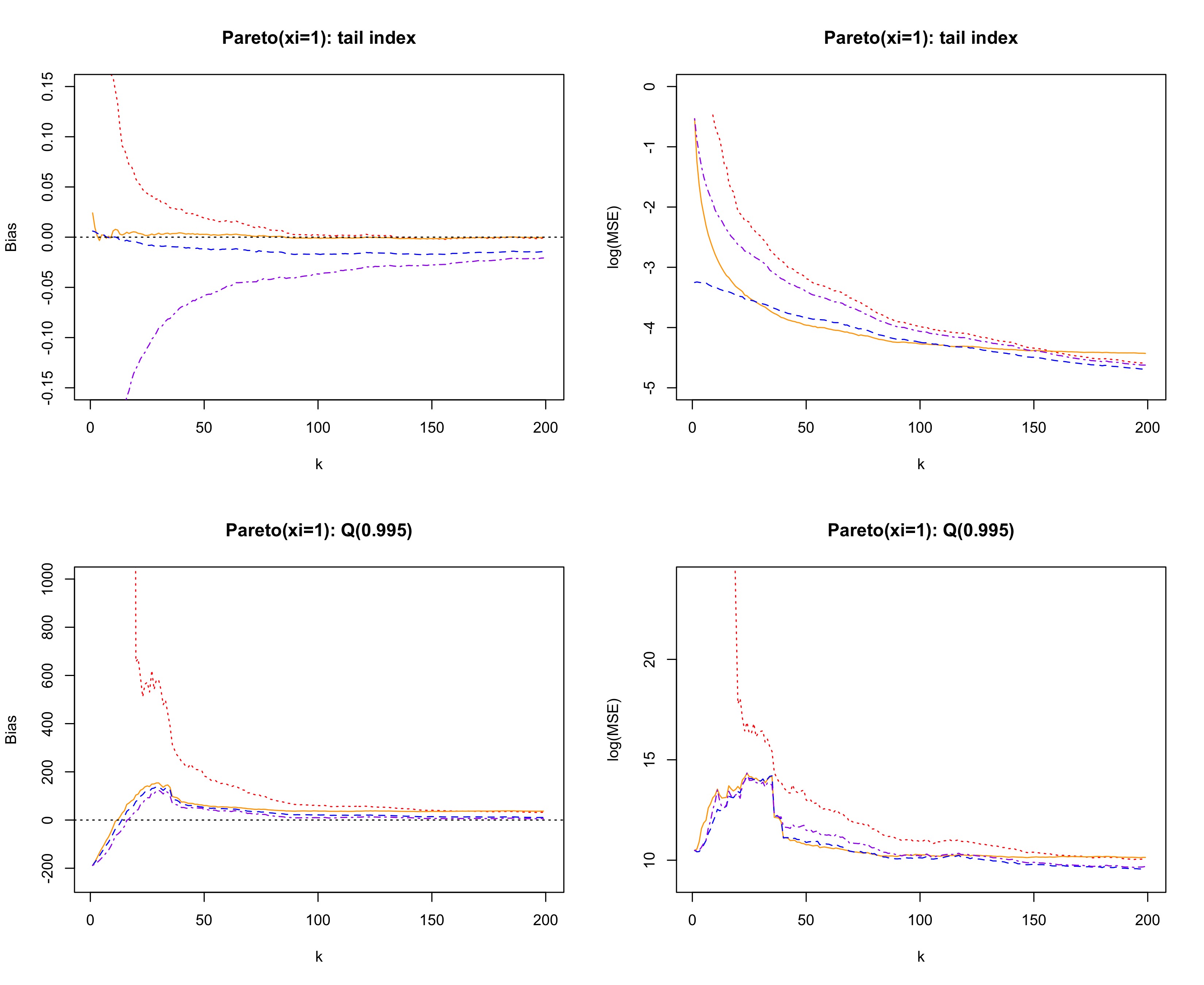}
\includegraphics[width=12cm,trim=0.5cm 0.5cm 0.5cm 0.5cm,clip]{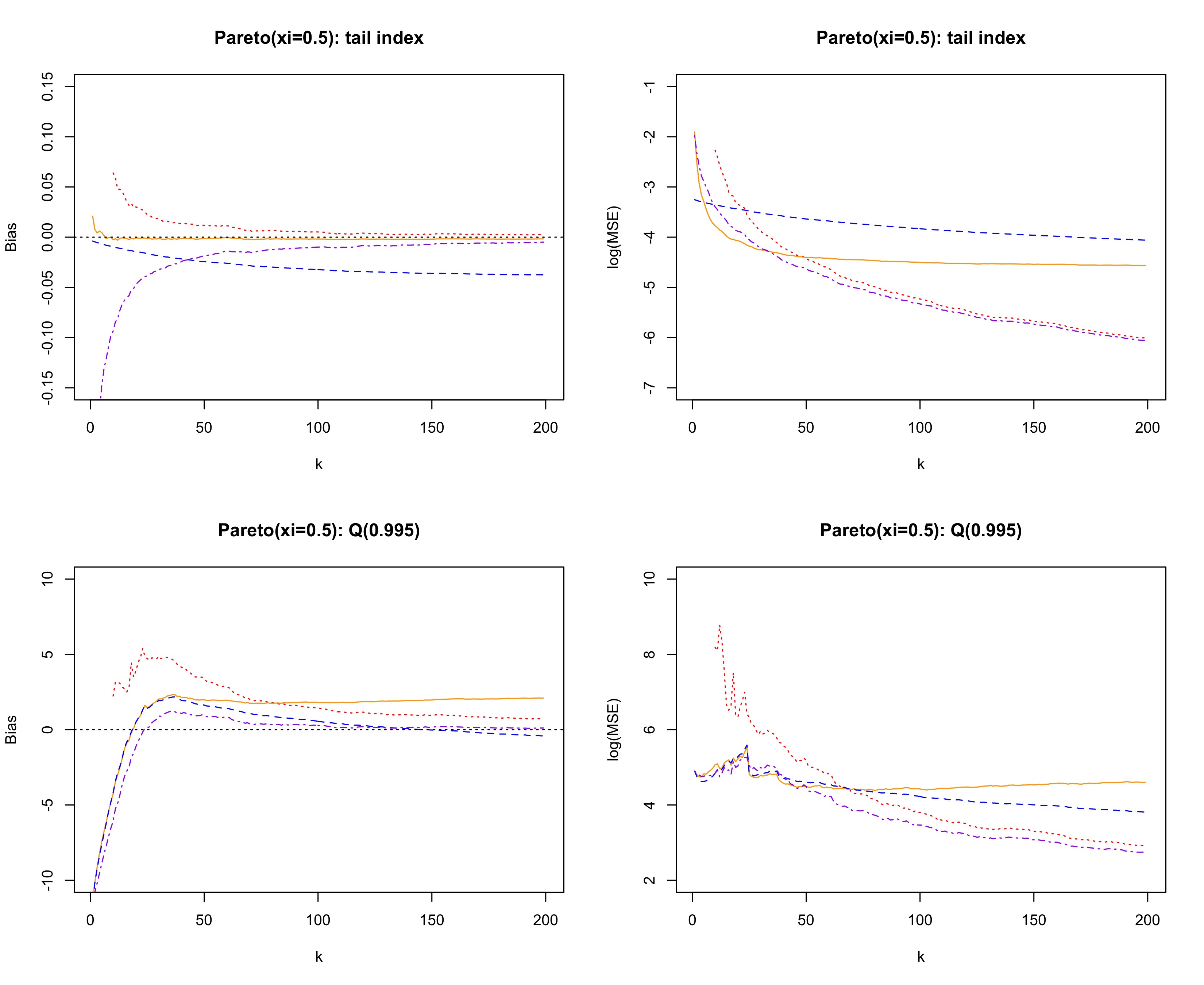}
\caption{Bias and (log) Mean Square Error for the exact Pareto distribution, for varying parameters. We compare $\hat\xi_k^P$ (orange, solid), $\hat\xi_k^{MLE}$ (red, dotted), $\hat\xi_k^{BG}$ (blue, dashed) and $\hat\xi_k^{BM}$ (purple, dashed and dotted), as well as the associated Weissman quantile estimator.} 
\label{sspareto}

\end{figure}
\begin{figure}[]
\centering
\includegraphics[width=12cm,trim=0.5cm 0.5cm 0.5cm 0.5cm,clip]{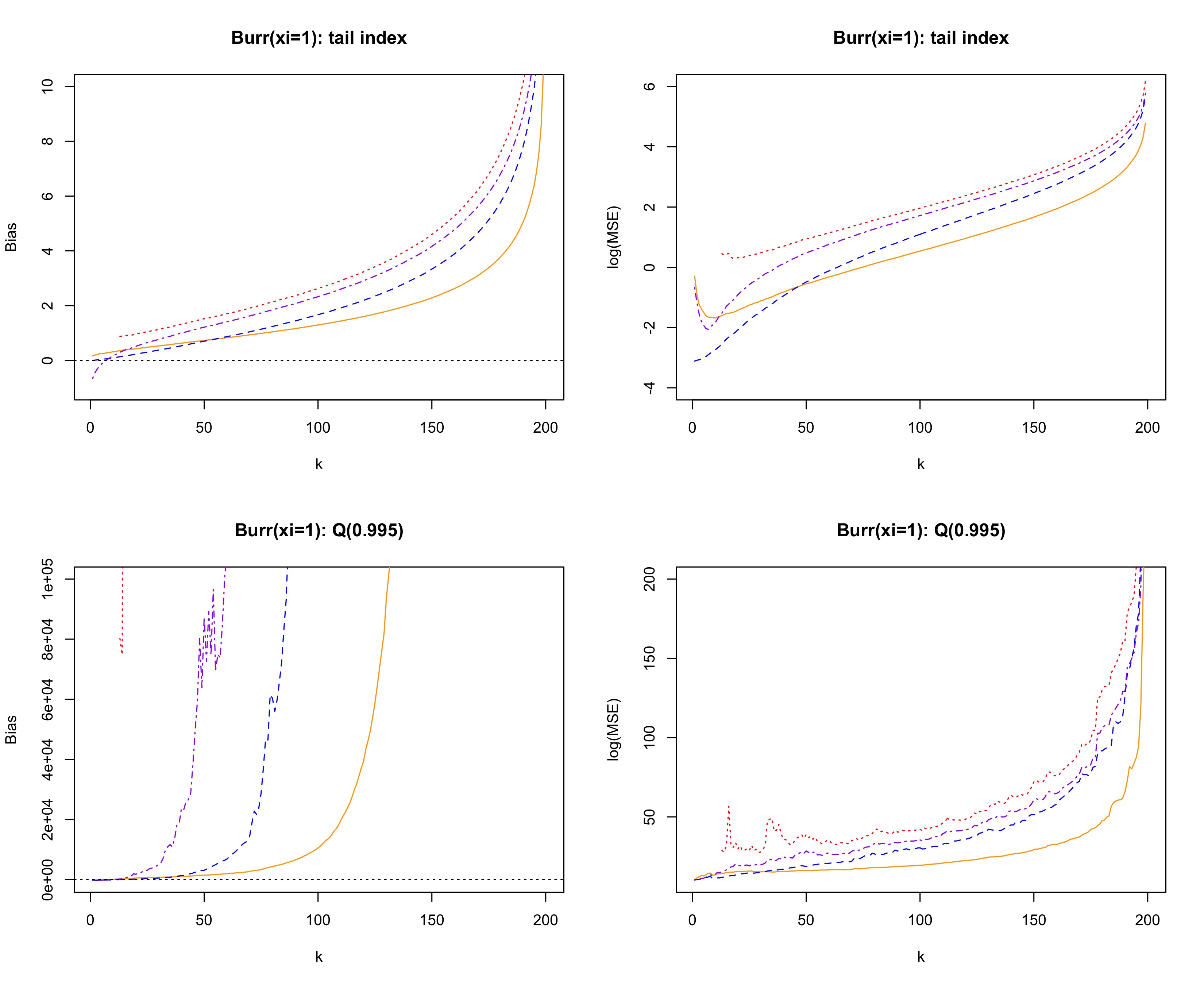}
\includegraphics[width=12cm,trim=0.5cm 0.5cm 0.5cm 0.5cm,clip]{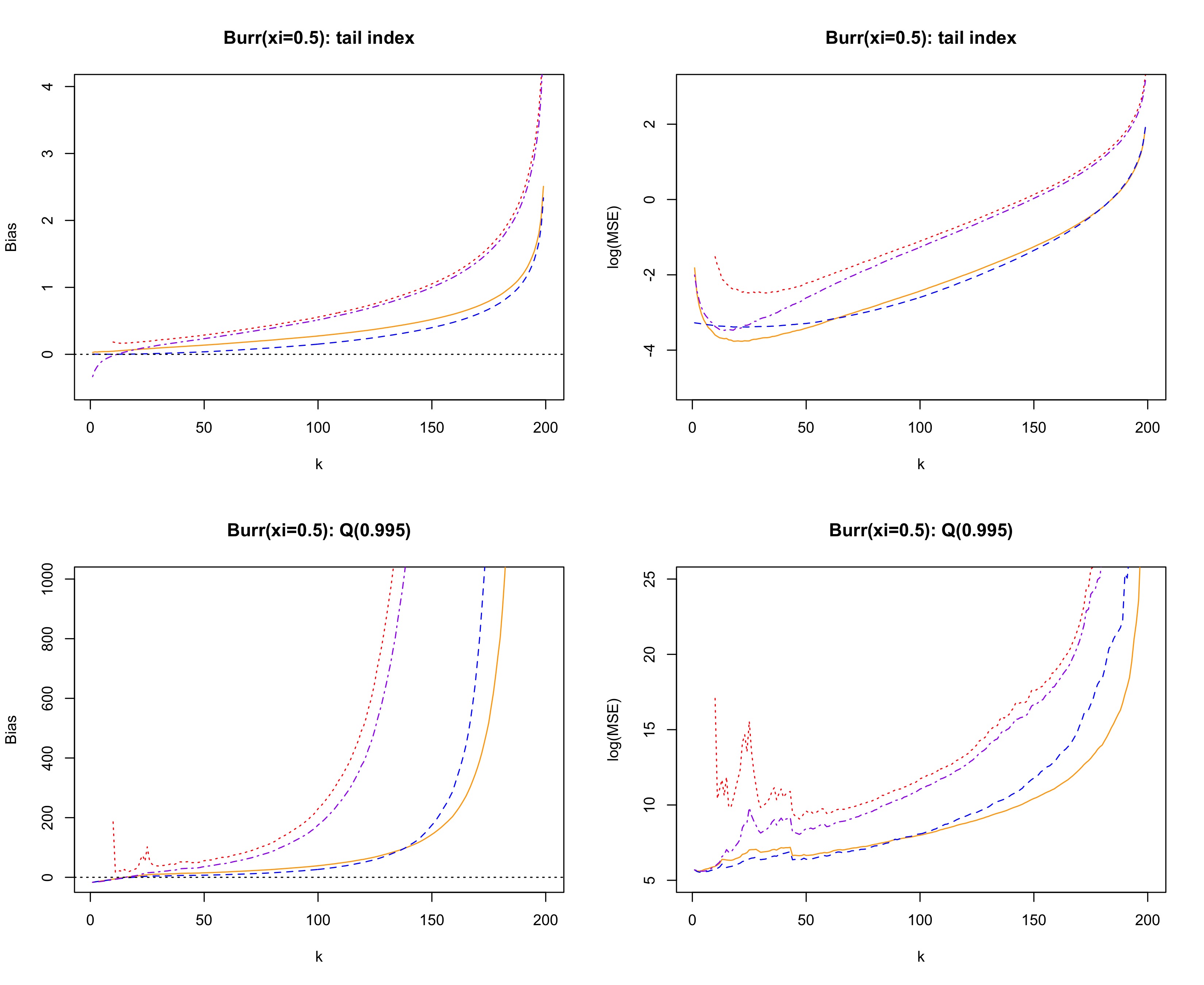}
\caption{Bias and (log) Mean Square Error for the Burr distribution, for varying parameters. We compare $\hat\xi_k^P$ (orange, solid), $\hat\xi_k^{MLE}$ (red, dotted), $\hat\xi_k^{BG}$ (blue, dashed) and $\hat\xi_k^{BM}$ (purple, dashed and dotted), as well as the associated Weissman quantile estimator.} 
\label{ssburr}
\end{figure}

\begin{figure}[]
\centering
\includegraphics[width=12cm,trim=0.5cm 0.5cm 0.5cm 0.5cm,clip]{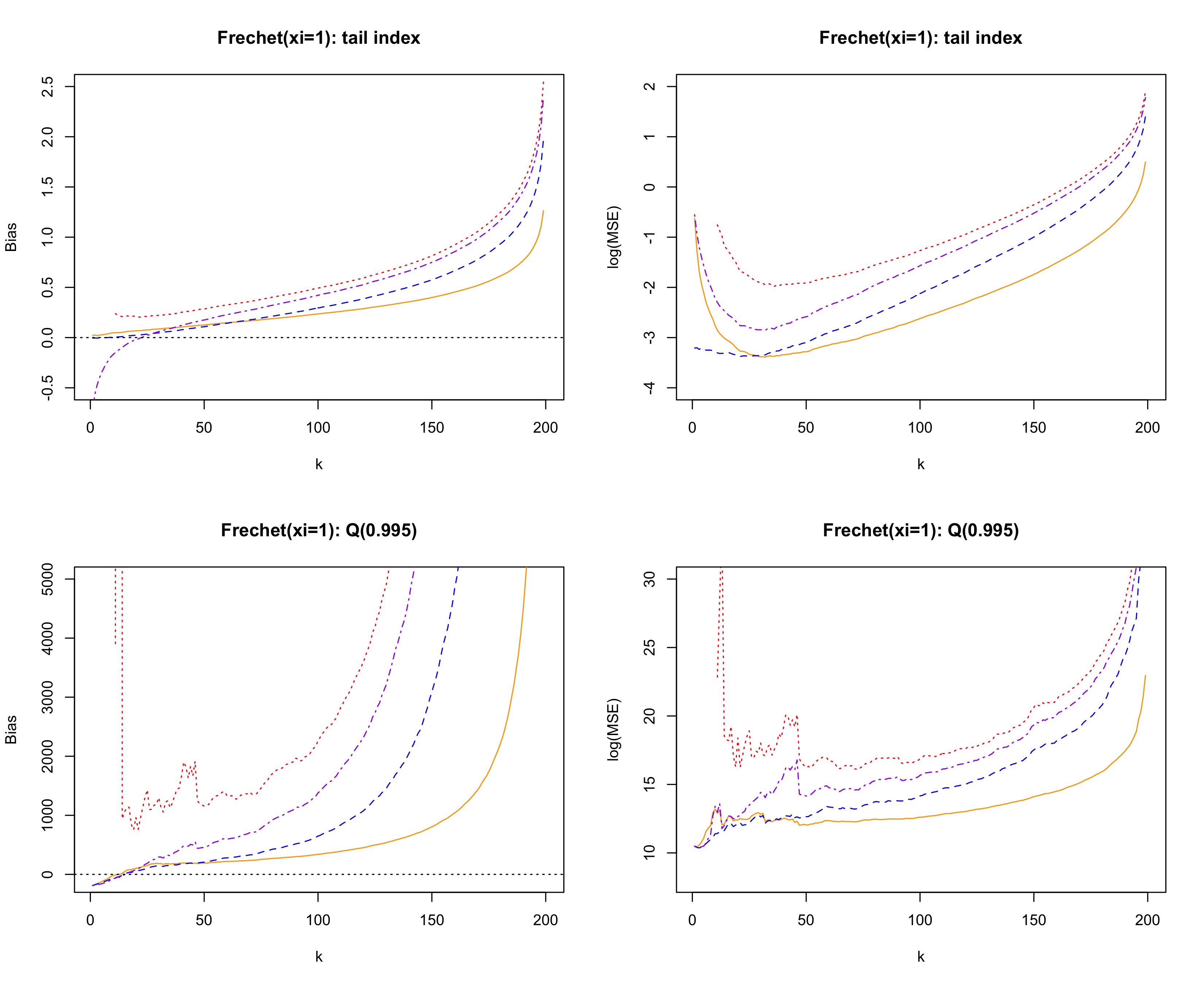}
\includegraphics[width=12cm,trim=0.5cm 0.5cm 0.5cm 0.5cm,clip]{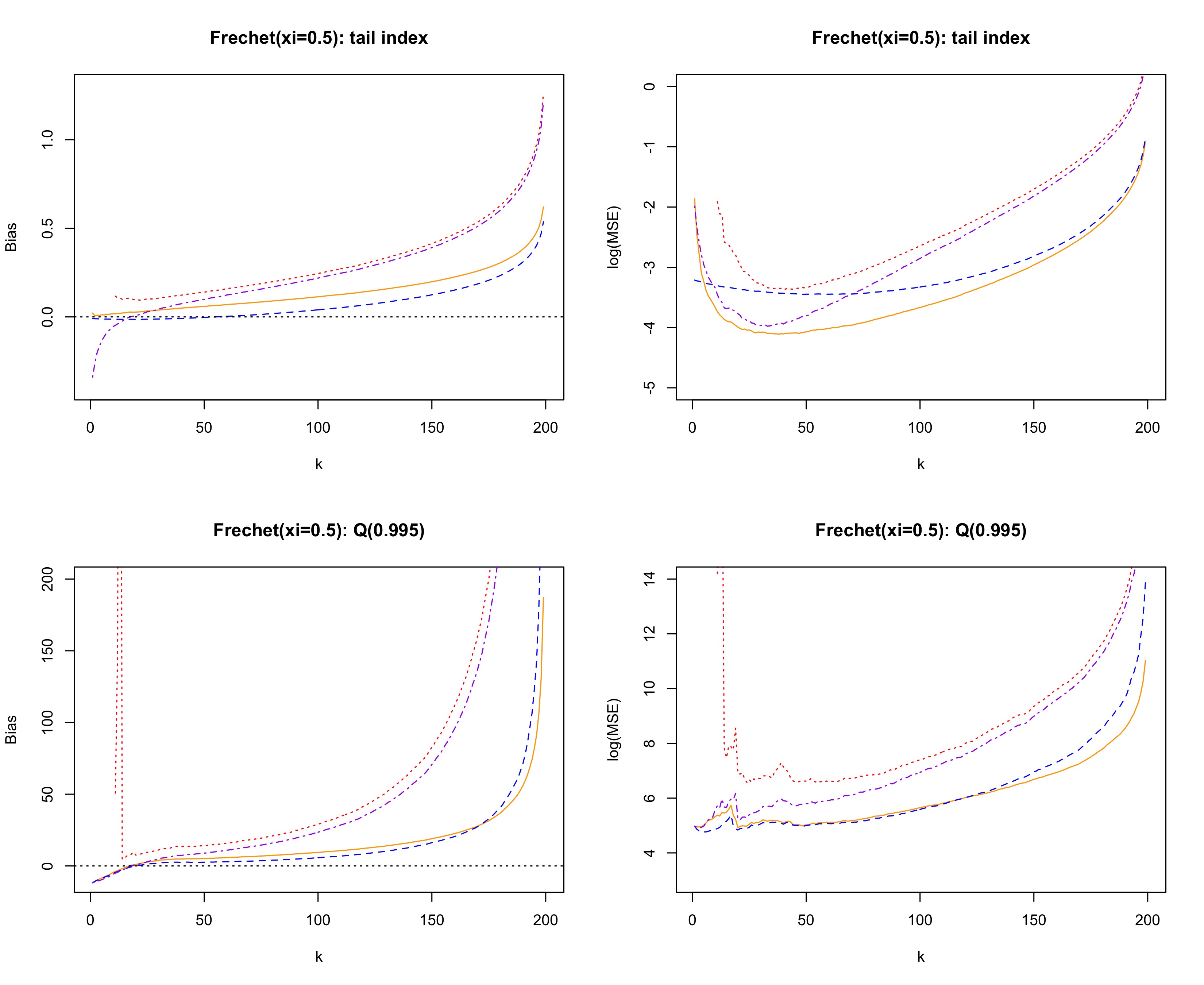}
\caption{Bias and (log) Mean Square Error for the Frechet distribution, for varying parameters. We compare $\hat\xi_k^P$ (orange, solid), $\hat\xi_k^{MLE}$ (red, dotted), $\hat\xi_k^{BG}$ (blue, dashed) and $\hat\xi_k^{BM}$ (purple, dashed and dotted), as well as the associated Weissman quantile estimator.} 
\label{ssfrechet}
\end{figure}

\subsection{Insurance Data}
We consider a dataset from Motor Third Party Liability Insurance (MTPL) from a direct insurance company operating in the EU (cf. \cite[Sec.1.3.1]{abt}), consisting of yearly paid amounts to policyholders during the period 1995-2010. At 2010 we have roughly $60\%$ right-censored (open) observations out of the total 837 claims. 
 The data are reported as soon as the incurred value exceeds the reporting threshold given in Figure 1.2 in \cite{abt}, and the histogram of the IBNR delays is given in Figure 1.3 in \cite{abt}.
We also have an \textit{ultimate} estimate which is the company's expert estimation of the eventual size of the claim. 
In Figure \ref{Description} we have several descriptive statistics of the data: the log-claim sizes, the Kaplan-Meier estimator of the data (cf. \cite{kaplan58}), the proportion of non-censoring (closed claims) as a function of the order statistics of the claims, and a QQ-plot for the log-claims against theoretical exponential quantiles. We observe that censoring at random is not a far-fetched assumption to make, since a horizontal behaviour of the proportion of closed claims as a function of the number of upper order statistics does not reject the possibility of the sizes of claims and the probabilities of censoring of those claims being independent. The Pareto tail behaviour of large (above 0.45 million, possibly censored) claim sizes seems to hold, based on the QQ-plot of their logarithm against theoretical exponential quantiles. Standard tests also do not reject the exponential hypothesis of the logarithm of these large claims (Kolmogorov-Smirnov p-value of 0.50). 

\begin{figure}[hh]
\centering
\includegraphics[width=13cm,trim=0.5cm 0.5cm 0.5cm 0.5cm,clip]{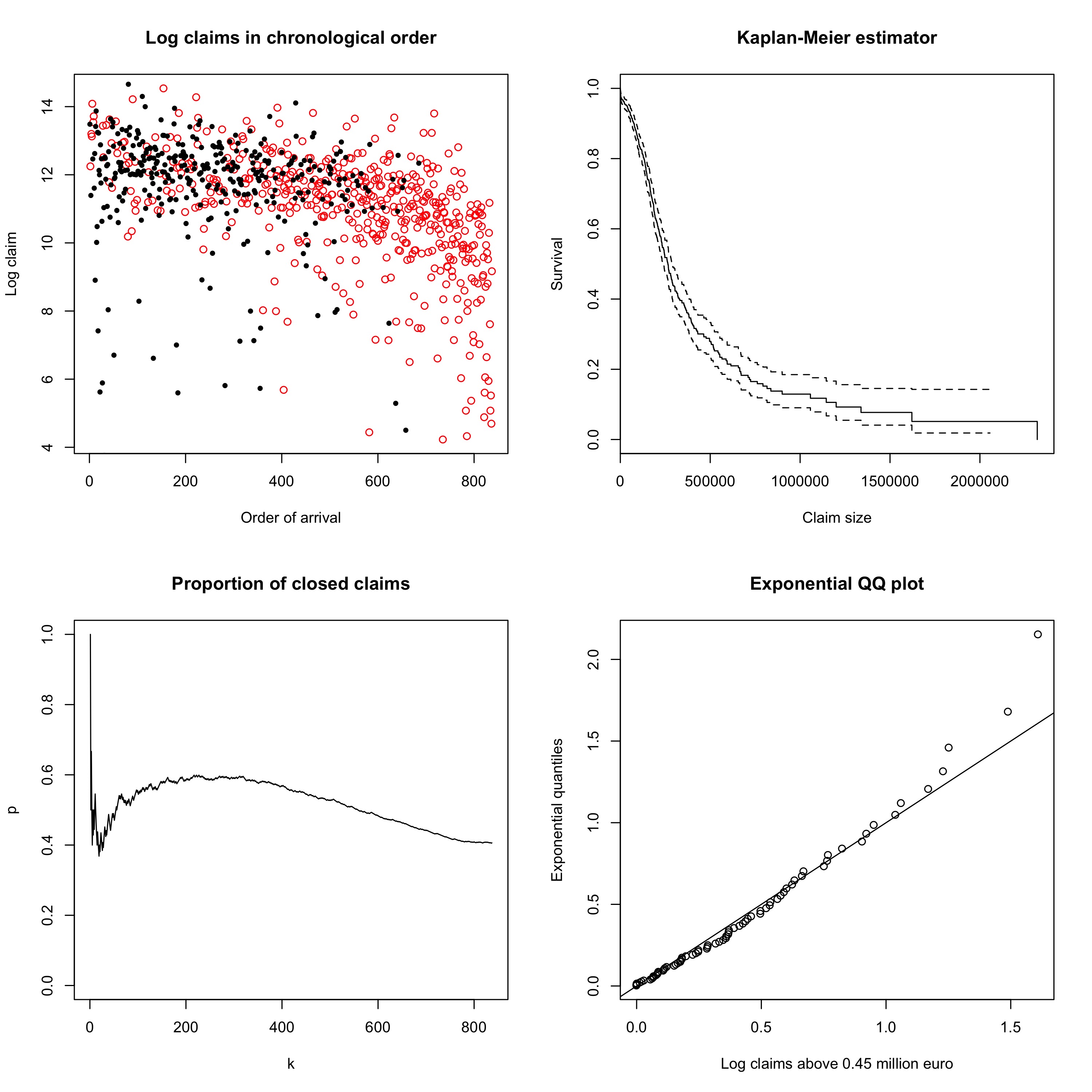}
\caption{Descriptive statistics of the insurance data. Top left: log-claims in order of arrival, showing both open (red, circle) and closed (black, dot) claims. Top right: Kaplan-Meier survival probability estimator for the claims. Bottom left: proportion of closed claims as a function of the top $k$ order statistics of the claim sizes. Bottom right: QQ plot of the logarithm of the claims larger than 0.54 million euro, against the theoretical exponential quantiles with the same mean.} 
\label{Description}
\end{figure}

Now we would like to know how the ultimates can help to estimate the tail parameter. In \cite{bladt}, the ultimates of this dataset were explored, where it was observed that they are Pareto in the tail. Furthermore, using developments in threshold selection, using trimming techniques, it was shown that $\xi=0.48$ is a good estimate for the heaviness of the tail of the ultimates. In Figure \ref{Hill_plot} we show the Hill plot for the ultimates, together with the chosen expert $\xi$ value, and the censored Hill $\hat\xi_k^{MLE}$ and the perturbed version $\hat\xi_k^P$ with $\lambda=1$ and $\beta=1/0.48$. Notice that in this case, we know how $\beta$ is obtained, and this additional knowledge could be useful. However, our method does not assume any specific structure, which means that any other method can be used to obtain $\beta$, and we merely give the current one for the sake of example. We observe a particularly stable region when $k$ is between $20$ and $70$, which suggests a heavier tail (roughly $0.65$) than the ultimates alone predict. We will see how this affects the quantiles alike.

\begin{figure}[]
\centering
\includegraphics[width=12cm,trim=0.5cm 0.5cm 0.5cm 0.5cm,clip]{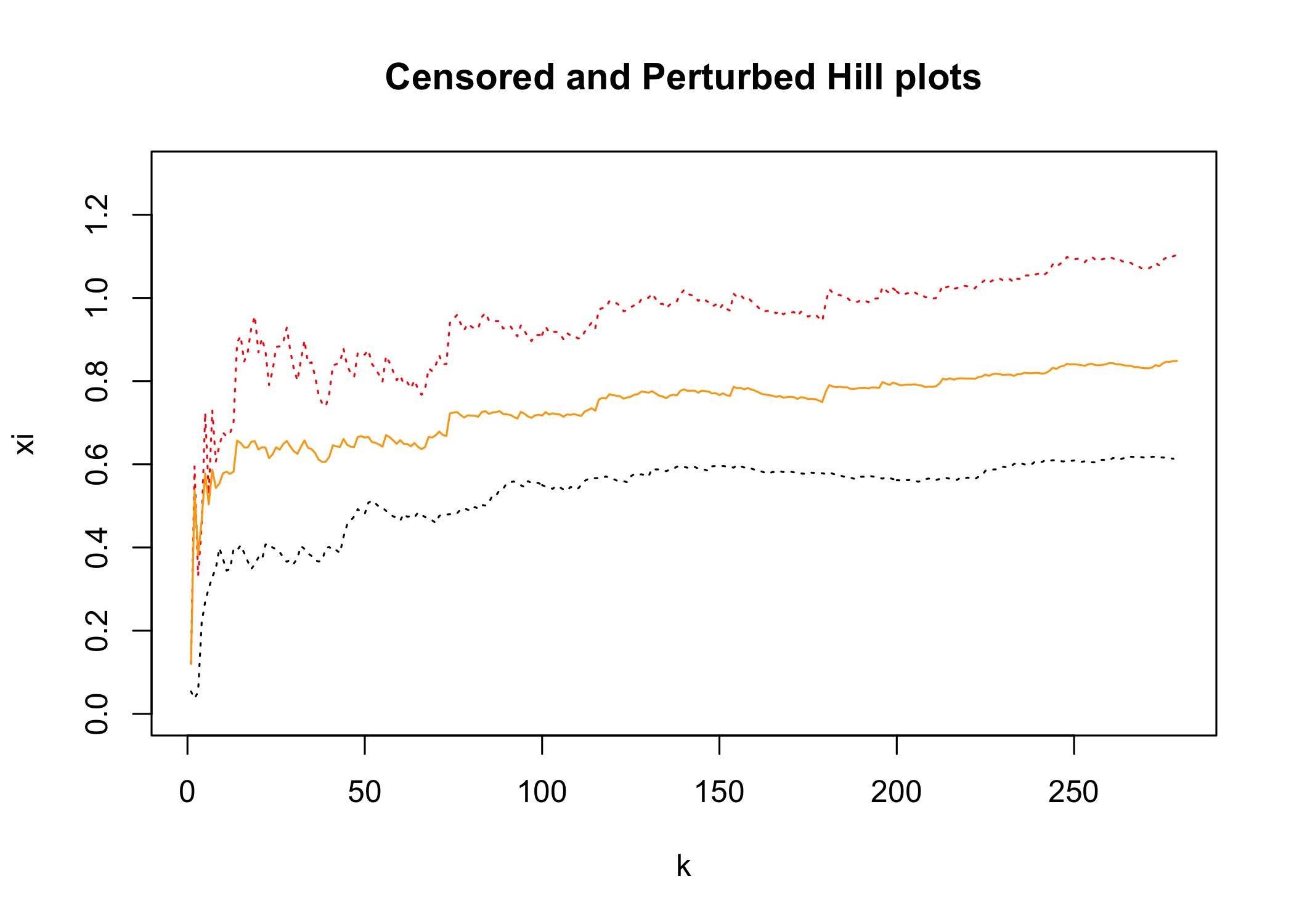}
\caption{Hill plot of the ultimates (black, dashed),
censored Hill estimator $\hat\xi_k^{MLE}$ for the claims (red, dashed), and the combined estimator $\hat\xi_k^P$ (orange, solid) with $\lambda=1$ and $\beta=1/0.48$.} 
\label{Hill_plot}
\end{figure}

As a way of validating our estimation procedure we perform the following check. We consider all claims arriving in the shorter period of time $1995$-$2000$ and we follow exclusively these $310$ older claims until $2010$. The proportion of censoring at $2010$ drops to roughly $29.5\%$. We examine the censored Hill estimator and perturbed estimator (using the same $\lambda=1$ and $1/\beta=0.48$ as before) for this reduced data, and plot it in Figure \ref{Hill_plot_verif}, together with the corresponding estimators using the full data which we had previously obtained. We observe that the censored Hill estimator for the reduced data dropped its value in the most stable region by about $0.2$, almost reaching the perturbed estimator for both the complete and reduced datasets, showing that as the proportion of censoring decreases, the estimators come closer together. Notice that the perturbed estimator remained surprisingly stable, even when the penalization parameter stayed at the same value but the sample size decreased, due to the fact that the proportion of censored claims controlled the strength of the penalization in a natural way.

Finally, we add the corresponding analysis of the $99.5\%$ quantile (which is relevant for Value-at-Risk considerations) for the case where the expert quantile information is given by the empirical distribution function of the ultimates, and combine the Hill estimator and the expert information by means of Equation \eqref{perturbed_quantile}. That is,

\begin{align*}
\hat Q^{P}_k(1-p)&=\left[\hat Q^{KM}(1-k/n)\cdot\left(\frac{k}{np}\right)^{\hat \xi^{MLE}_k}\right]^{\hat p_k}\cdot \left[\hat Q^{ULT}(1-k/n)\cdot\left(\frac{k}{np}\right)^{1/\beta}\right]^{1-\hat p_k}\\
&=\left((\hat Q^{KM}(1-k/n))^{\hat p_k}(\hat Q^{EX}(1-k/n))^{1-\hat p_k}\right)^{\hat\xi^P_k},
\end{align*}
where $\hat p_k=\frac 1k \sum_{i=1}^k \ee^{(i)}$, $\hat Q^{KM}$ is the quantile function associated with the Kaplan-Meier curve of the claims, and $\hat Q^{ULT}$ is the quantile function associated with the empirical distribution function of the ultimates.

\noindent The quantile coming from the ultimates alone is given by
\begin{align*}
\hat Q^{ULT}_k(1-p)=\hat Q^{ULT}(1-k/n)\cdot\left(\frac{k}{np}\right)^{H^U_k},
\end{align*}
where $H^U_k$ is the Hill estimator of the ultimates. Finally, without any expert information (ignoring the ultimates), the quantile is given by
\begin{align*}
\hat Q^{KM}_k(1-p)=\hat Q^{KM}(1-k/n)\cdot\left(\frac{k}{np}\right)^{\hat \xi^{MLE}_k}.
\end{align*}
 
\noindent  Note that, due to missing IBNR data at the later accident years, some care is needed concerning the interpretation of these quantile estimates as the outcome levels which are exceeded in $100 \times 0.5 \%$ {\it of the reported cases.} However, as these IBNR data concern smaller losses, the influence of these omissions is limited as can be verified by restricting the proposed approach to the claims from earlier accident years and comparing with the present results. The result is gathered as a function of the number $k$ of upper order statistics in Figure \ref{VaR}. The combined estimation of the high quantile results is a stable compromise between the under-estimated quantiles from the expert opinions and the pure Weissman approach, which has higher variability. Such under-estimation of the size of the claims at closure by the ultimates was also observed empirically while exploring the data (details are omitted). Observe also Figure \ref{VaR_verif}, where the reduction of the data which was used above to validate the procedure was applied to the quantiles, and an analogous interpretation applies. This suggests that the current reserving could benefit from a re-evaluation. However, this analysis is made without knowing the actual process behind the calculation of the ultimates, and a deeper understanding of this process could in the future elucidate whether there is something being overlooked by the experts or by the statisticians.

\begin{figure}[]
\centering
\includegraphics[width=12cm,trim=0.5cm 0.5cm 0.5cm 0.5cm,clip]{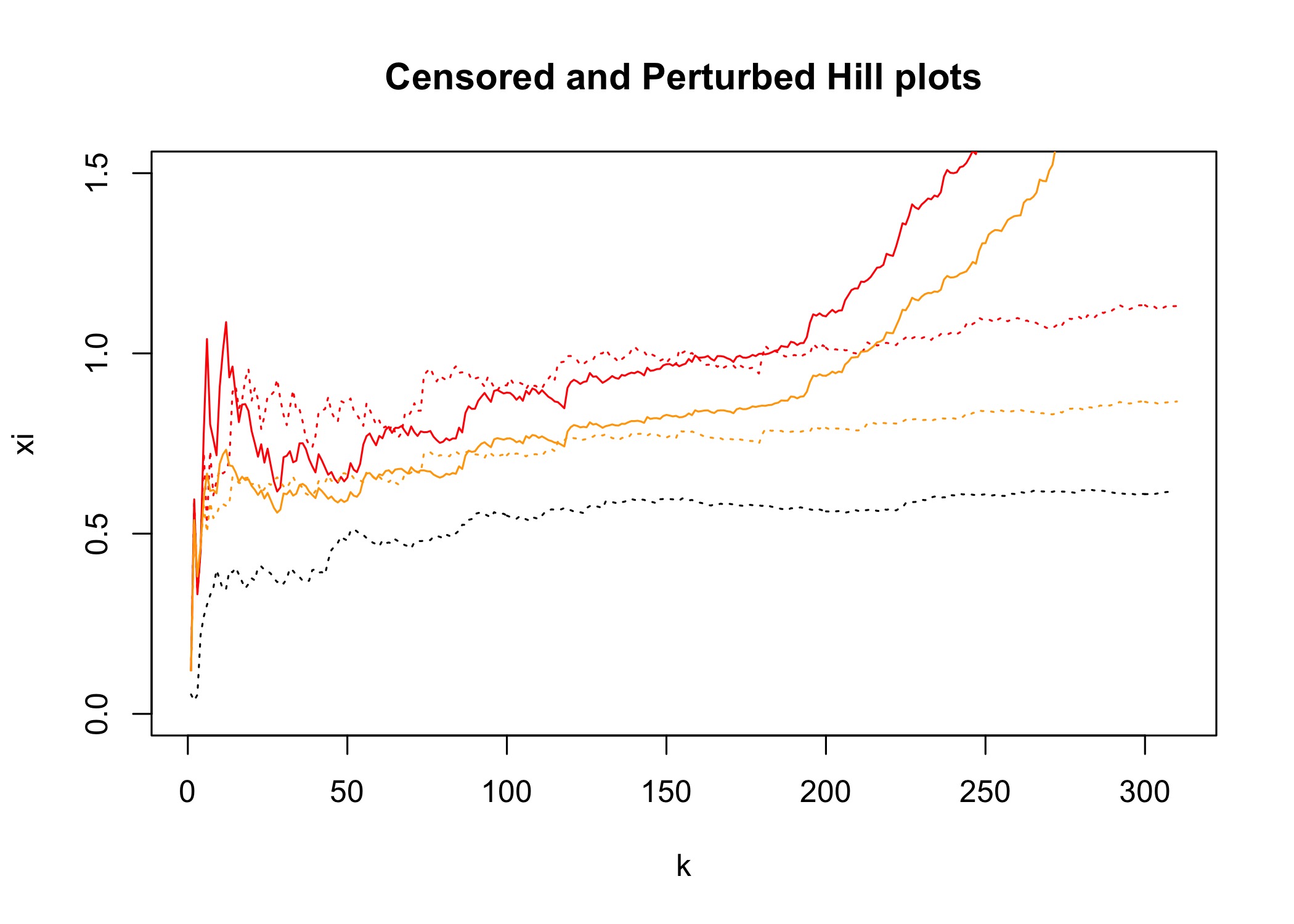}
\caption{Hill plot of the ultimates (black), for the reduced (solid) and complete (dashed) datasets: 
censored Hill estimator $\hat\xi_k^{MLE}$ for the claims (red), and the combined estimator $\hat\xi_k^P$ (orange) with $\lambda=1$ and $\beta=1/0.48$.} 
\label{Hill_plot_verif}
\end{figure}

\begin{figure}[]
\centering
\includegraphics[width=12cm,trim=.5cm .5cm 0cm .5cm,clip]{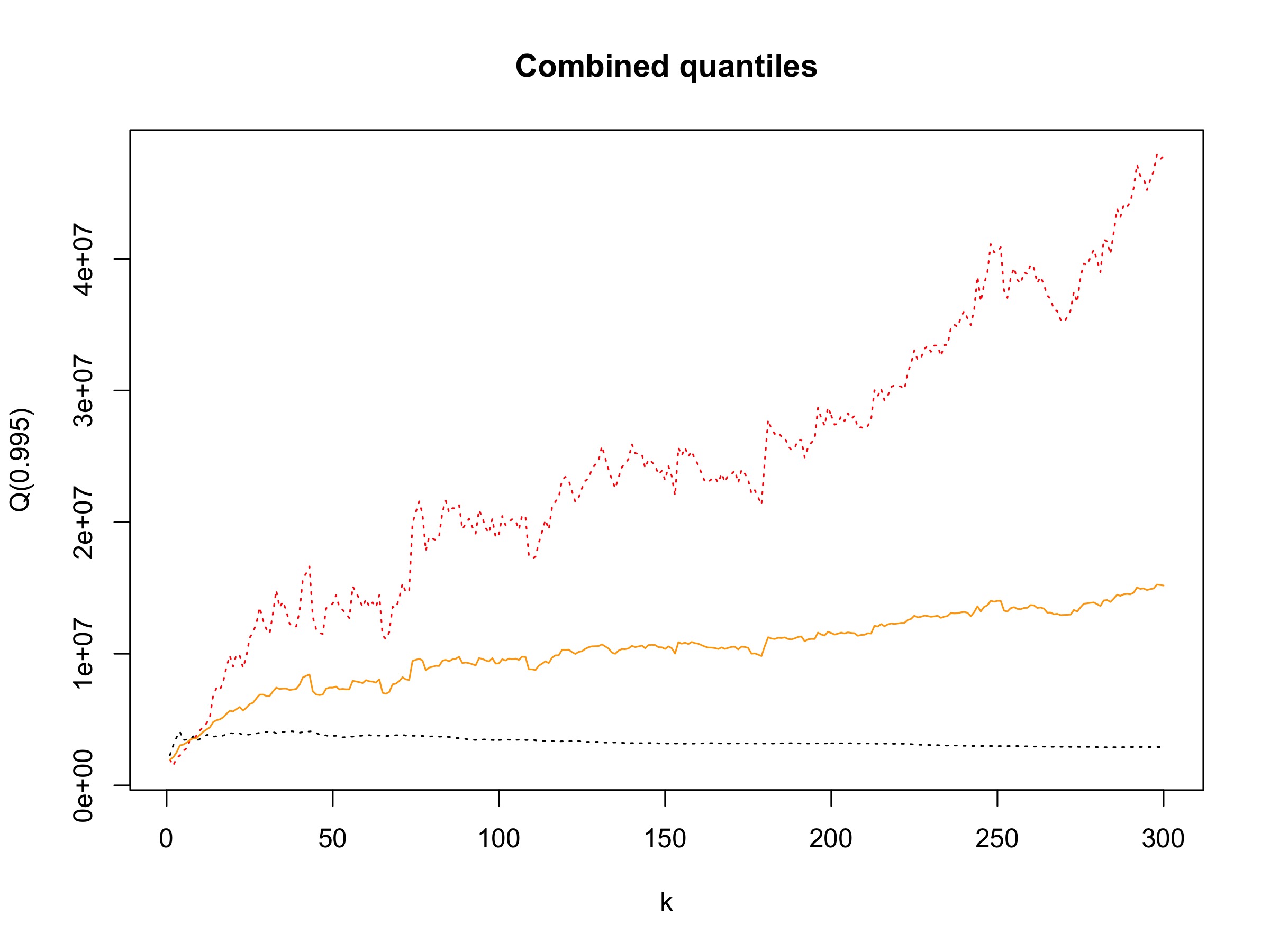}
\caption{$99.5\%$ quantile estimator using the censored approach ($\hat Q_k^{KM}(0.005)$, red) for the claims, expert information ($\hat Q_k^{ULT}(0.005)$, black) and their combination via $\hat \xi^{P}_u$, with the selection $\lambda=1$ ($\hat Q^P_k(0.005)$, orange).} 
\label{VaR}
\end{figure}

\begin{figure}[]
\centering
\includegraphics[width=12cm,trim=.5cm .5cm 0cm .5cm,clip]{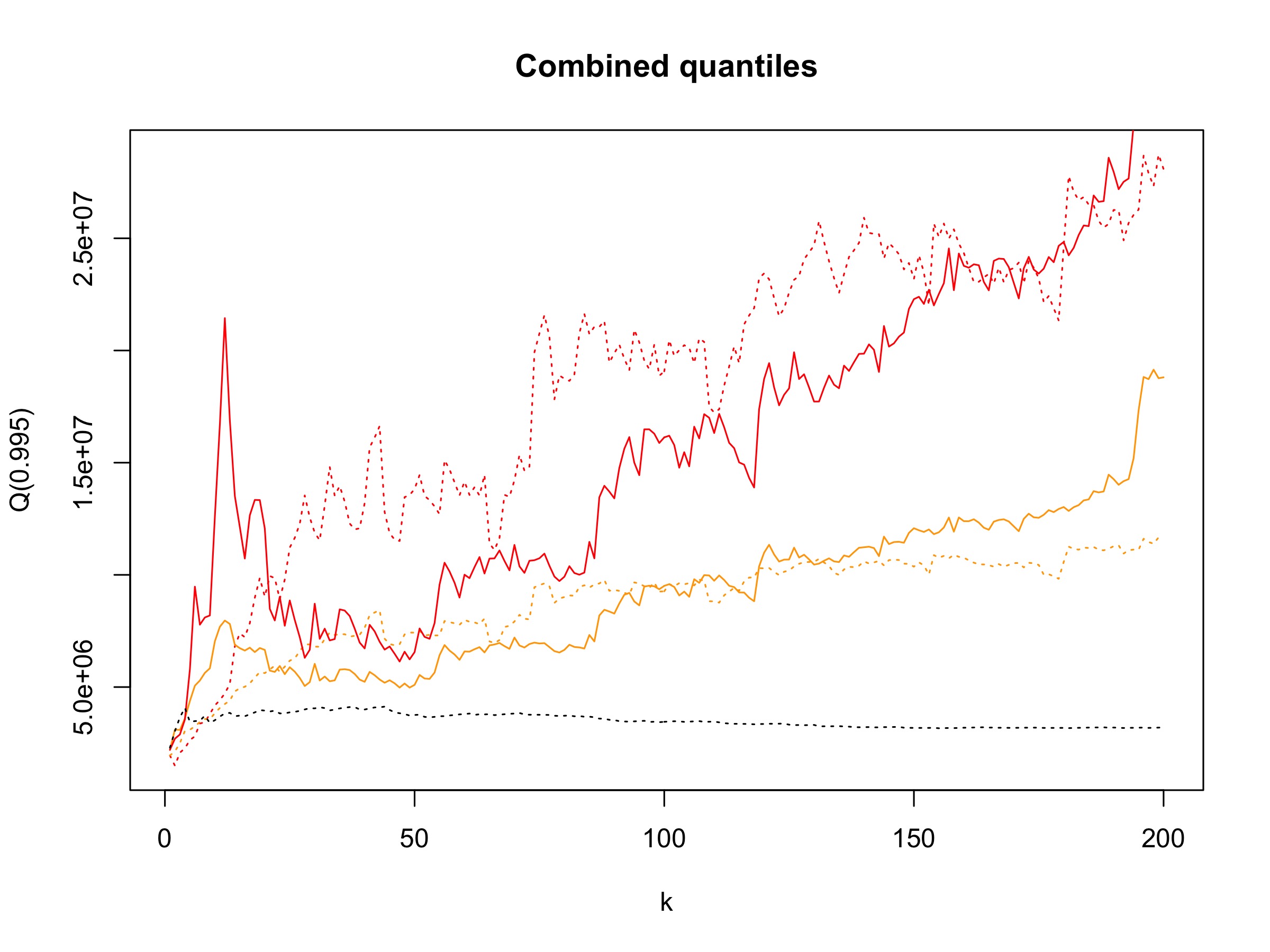}
\caption{
Plot of the $99.5\%$ quantile for the ultimates $\hat Q^{ULT}_k$ (black), and for the reduced (solid) and complete (dashed) datasets: 
censored Hill estimation $\hat Q_k^{KM}$ for the claims (red), and the combined estimator $\hat Q_k^P$ (orange) with $\lambda=1$ and $\beta=1/0.48$.} 
\label{VaR_verif}
\end{figure}

\section{Conclusion}\label{conclusion_sect}
We have derived a flexible estimator that bridges statistical theory and practice when it comes to tail estimation. The results also apply to adaptation of quantile estimation techniques both when more expert information is available (for instance when an expert cumulative distribution function is available) and also when it is lacking. Like in Bayesian statistics, the strength of the belief of the expert is often subjective and in many cases unquantifiable, especially when provided with a single point estimate. As discussed in the paper, our method is in fact closely related to Bayesian techniques, but it is driven by the proportion of censoring, rather than by the number of total observations. The developed estimator represents a statistically sound method for making a compromise between expert information and likelihood methods, without the need of any additional prior assumptions, and its  performance depends on the quality of the expert guess. In particular, we suggested a convenient approach to avoid selection of a tuning parameter for the linking of expert information and Hill estimation. The methods developed can readily be adapted for the selection of the tuning parameter using more complex methods (such as moment matching) whenever there is more expert information available than presently assumed.

For heavy tails, the estimator is shown to be asymptotically normal, and has further desirable properties when the tuning parameter is chosen to be $1$. Indeed, Theorem \ref{lambdaone} can serve as a simple rule of thumb in practice for combining the two sources of information, and suggests that using good quality expert information can reduce the variance while keeping the bias at bay. This rule appears to be rather natural, and the approach in this paper enables to embed this intuitive combination within the theory of perturbed likelihood estimation. A more detailed analysis would depend on the specific application at hand, and on the quantifiability of the strength of beliefs, which in the present liability insurance dataset, and more generally in any analysis made by statisticians without the experts present, is commonly lacking.

A simulation study showed that when the guess is close to being correct, the estimator fares very favorably, compared to the Hill estimator and two recently proposed Bayesian solutions. Moreover, the estimator seems to be quite stable with respect to the chosen threshold, which is of particular interest since the choice of an appropriate threshold is a classical problem in extreme value analysis. Concerning quantiles, and for the simulated examples, the estimator was favorable to all the benchmarks for virtually all sample fractions for non-exact Pareto tails. Trimming techniques  have recently been proposed to address threshold selection, and a future line of research will be to consider lower-trimmed versions of the proposed perturbed estimator to aid in the visual and automatic sample fraction selection.

 Finally, the application of the method to actual motor third-party liability liability insurance data illustrates that decision makers with a strong belief in a point estimate of the tail parameter could be less reluctant to use the tail parameter and quantiles suggested by the inclusion of data-points proposed by our method than the one from the pure censored Hill estimation of the data. \\
 
 \textbf{Acknowledgement.} H.A. acknowledges financial support from the Swiss National Science Foundation Project 200021\_168993.

\appendix
\section{Proof of Theorem \ref{biasvar}}
\begin{proof}
 Define
\begin{align}
  V_k=\frac{1}{k}\sum_{i=1}^k(1-\ee^{(i)})(1-\lambda ),\quad  W_k=\frac{1}{k}\sum_{i=1}^k(1-\ee^{(i)})/\beta.
\end{align}
First note that
\[
 V_k\stackrel{d}{\to} r_1, \; 
 W_k \stackrel{d}{\to}_p r_2.
\]
Concerning the asymptotic distribution of $V_k$ and $W_k$ we make use of the method developed in \cite{einmahl2008statistics} introducing i.i.d.\ uniform (0,1) random variables $U_i$, $i \geq 1$, independent of the $Z_i$ sequence, and corresponding indicators being equal to
$1$ if $U \leq p(Z)$, and 0 otherwise. When denoting the $U$ variable induced by  $Z^{(i)}$ by $U^{(i)}$ we have
\[
V_k\stackrel{d}{=} {1 \over k}\sum_{j=1}^k 1_{ \{U^{(i)} > p (Z^{(i)})\} }
(1-\lambda)
\mbox{ and } 
W_k\stackrel{d}{=}{1 \over {\beta k}}\sum_{j=1}^k 1_{\{U^{(i)} > p (Z^{(i)})\}}
,
\]
which, under \eqref{biaspU}, can be replaced asymptotically by 
\[
\hat{V}_k= {1 \over k}\sum_{j=1}^k 1_{ \{U^{(i)} > p (U(n/i)) \} }
(1-\lambda )
\mbox{ and } 
\hat{W}_k= {1 \over {\beta k}}\sum_{j=1}^k 1_{ \{U^{(i)} > p (U(n/i))\}}
,
\]
for which
\begin{align*}
\sqrt{k}(\hat{V}_k-r_1)\stackrel{d}{\to} Y_1,\quad \sqrt{k}( \hat{W}_k-r_2)\stackrel{d}{\to}Y_2,
\end{align*}
where


\begin{eqnarray*}
Y_1&\sim& \mathcal{N}(-\frac{\delta \nu_*\alpha_c}{\alpha_c+\nu_\ast} p(1-p)C^{-\nu_*/\alpha_c}(D/\alpha)_*(1-\lambda),\; p(1-p)(1-\lambda)^2),\\ Y_2&\sim& \mathcal{N}(-\frac{\delta \nu_*\alpha_c}{\alpha_c+\nu_\ast} p(1-p)C^{-\nu_*/\alpha_c}(D/\alpha)_*/\beta,\; p(1-p)/\beta^2),
\end{eqnarray*}
which are independent of $Y_0$, defined in \eqref{y0}.
Then we obtain
\begin{eqnarray*}
\sqrt{k}\left(\frac{ H_k+\lambda  \hat{W}_k}{1- \hat{V}_k}-\frac{1+\lambda  \alpha_2/\beta}{\alpha + \lambda \alpha_2} \right)
&=&
\frac{1}{1- \hat{V}_k}\sqrt{k}\left( H_k-\frac{1}{\alpha_c}\right)+
\frac{\lambda }{1-\hat{V}_k} \sqrt{k}( \hat{W}_k-r_2)\\
&&+\left( \lambda r_2+\frac{1}{\alpha_c}\right) \sqrt{k}\left( \frac{1}{1- \hat{V}_k}-\frac{1}{1-r_1}\right) \\
&=& \frac{1}{1- \hat{V}_k} \sqrt{k}\left( H_k-\frac{1}{\alpha_c}\right)+\frac{\lambda}{1- \hat{V}_k} \sqrt{k}( \hat{W}_k-r_2)\\
&&+\frac{\lambda r_2+\frac{1}{\alpha+\alpha_2}}{(1-r_1)(1- \hat{V}_k)} \sqrt{k}\left(\hat{V}_k-r_1\right) \\
&&\stackrel{d}{\to} \frac{1}{1-r_1} Y_0+\frac{\lambda}{1-r_1}Y_2+\frac{\lambda r_2 +\frac{1}{\alpha_c}}{(1-r_1)^2}Y_1\\
&&=\frac{1}{1-r_1} Y_0+\left[\frac{\lambda}{1-r_1}\frac{1}{\beta(1-\lambda)}+\frac{\lambda r_2 +\frac{1}{\alpha_c}}{(1-r_1)^2}\right]Y_1
\end{eqnarray*}
The mean and variance are then computed from the last expression, since $Y_0$ and $Y_1$ are independent.
\end{proof}

\bibliographystyle{plain}
\bibliography{Exponential_penalization5.bib}

\begin{thebibliography}{15}
\providecommand{\natexlab}[1]{#1}
\providecommand{\url}[1]{\texttt{#1}}
\expandafter\ifx\csname urlstyle\endcsname\relax
  \providecommand{\doi}[1]{doi: #1}\else
  \providecommand{\doi}{doi: \begingroup \urlstyle{rm}\Url}\fi

\bibitem[Albrecher et~al.(2017)Albrecher, Teugels, and Beirlant]{abt}
Hansj{\"o}rg Albrecher, Jozef~L Teugels, and Jan Beirlant.
\newblock \emph{Reinsurance: Actuarial and Statistical Aspects}.
\newblock John Wiley \& Sons, 2017.

\bibitem[Beirlant et~al.(2007)Beirlant, Dierckx, Guillou, and
  Fils-Villetard]{beirlcens2007}
Jan Beirlant, Goedele Dierckx, Armelle Guillou, and Amelie Fils-Villetard.
\newblock Bias reduced tail estimation for censored pareto type distributions.
\newblock \emph{Extremes}, 10:\penalty0 151--174, 2007.

\bibitem[Einmahl et~al.(2008)Einmahl, Fils-Villetard, Guillou,
  et~al.]{einmahl2008statistics}
John~HJ Einmahl, Am{\'e}lie Fils-Villetard, Armelle Guillou, et~al.
\newblock Statistics of extremes under random censoring.
\newblock \emph{Bernoulli}, 14\penalty0 (1):\penalty0 207--227, 2008.

\bibitem[Worms and Worms(2014)]{worms2014new}
Julien Worms and Rym Worms.
\newblock New estimators of the extreme value index under random right
  censoring, for heavy-tailed distributions.
\newblock \emph{Extremes}, 17\penalty0 (2):\penalty0 337--358, 2014.

\bibitem[Ameraoui et~al.(2016)Ameraoui, Boukhetala, and
  Dupuy]{ameraoui2016bayesian}
Abdelkader Ameraoui, Kamal Boukhetala, and Jean-Fran{\c{c}}ois Dupuy.
\newblock Bayesian estimation of the tail index of a heavy tailed distribution
  under random censoring.
\newblock \emph{Computational Statistics \& Data Analysis}, 104:\penalty0
  148--168, 2016.

\bibitem[Beirlant et~al.(2018)Beirlant, Maribe, and
  Verster]{beirlant2018penalized}
Jan Beirlant, Gaonyalelwe Maribe, and Andrehette Verster.
\newblock Penalized bias reduction in extreme value estimation for censored
  pareto-type data, and long-tailed insurance applications.
\newblock \emph{Insurance: Mathematics and Economics}, 78:\penalty0 114--122,
  2018.

\bibitem[Bogaerts et~al.(2018)Bogaerts, Komarek, and Lesaffre]{lesaffre}
K.~Bogaerts, A.~Komarek, and E.~Lesaffre.
\newblock \emph{Survival analysis with interval-censored data}.
\newblock Chapman \& Hall, Boca Raton, 2018.

\bibitem[Hill(1975)]{Hill}
Bruce~M Hill.
\newblock A simple general approach to inference about the tail of a
  distribution.
\newblock \emph{The Annals of Statistics}, 3:\penalty0 1163--1174, 1975.

\bibitem[Embrechts et~al.(2013)Embrechts, Kl{\"u}ppelberg, and
  Mikosch]{embrechts2013modelling}
Paul Embrechts, Claudia Kl{\"u}ppelberg, and Thomas Mikosch.
\newblock \emph{Modelling extremal events: for insurance and finance},
  volume~33.
\newblock Springer Science \& Business Media, second edition, 2013.

\bibitem[Beirlant et~al.(2004)Beirlant, Goegebeur, Segers, and
  Teugels]{BGST2004}
Jan Beirlant, Yuri Goegebeur, Johan Segers, and Jozef Teugels.
\newblock \emph{Statistics of Extremes: Theory and Applications}.
\newblock Wiley, 2004.

\bibitem[Leung et~al.(1997)Leung, Elashoff, and Afifi]{leung1997censoring}
Kwan-Moon Leung, Robert~M Elashoff, and Abdelmonem~A Afifi.
\newblock Censoring issues in survival analysis.
\newblock \emph{Annual review of public health}, 18\penalty0 (1):\penalty0
  83--104, 1997.

\bibitem[Hall(1982)]{hall82}
Peter Hall.
\newblock On some simple estimates of an exponent of regular variation.
\newblock \emph{J. Roy. Statist. Soc. Ser. B}, 44\penalty0 (1):\penalty0
  37--42, 1982.

\bibitem[Weissman(1978)]{weissman1978estimation}
Ishay Weissman.
\newblock Estimation of parameters and large quantiles based on the k largest
  observations.
\newblock \emph{Journal of the American Statistical Association}, 73\penalty0
  (364):\penalty0 812--815, 1978.

\bibitem[Kaplan and Meier(1958)]{kaplan58}
E.~L. Kaplan and Paul Meier.
\newblock Nonparametric estimation from incomplete observations.
\newblock \emph{Journal of the American Statistical Association}, 53\penalty0
  (282):\penalty0 457--481, 1958.

\bibitem[Bladt et~al.(2019)Bladt, Albrecher, and Beirlant]{bladt}
Martin Bladt, Hansjoerg Albrecher, and Jan Beirlant.
\newblock Trimming and threshold selection in extremes.
\newblock \emph{arXiv:1903.07942}, 2019.

\end{thebibliography}

\end{document}